\documentclass[a4paper,
11pt,
DIV 11]{scrartcl}
\usepackage{url}
\newcommand{\email}[1]{\url{#1}}
\usepackage[english]{babel}
\usepackage[utf8x]{inputenc}% WICHTIG! Für rôle etc.
\usepackage{xcolor}
\usepackage{amssymb,amsmath,amsthm}
\usepackage{tikz}
\usetikzlibrary{arrows,positioning,backgrounds,fit}
\usepackage[
 plainpages=false,
 pdftitle={Automated Complexity Analysis Based on the Dependency Pair Method},
 pdfauthor={Nao Hirokawa and Georg Moser},
 pdfkeywords={
   Termination Methods, Runtime
  Complexity, Dependency Pair Method},
]{hyperref}
\usepackage{arxiv}

\theoremstyle{plain}
\newtheorem{theorem}{Theorem}[section]
\newtheorem{corollary}[theorem]{Corollary}
\newtheorem{lemma}[theorem]{Lemma}

\newtheorem*{claim*}{Claim}

\newtheorem{proposition}[theorem]{Proposition}
\theoremstyle{definition}      % Definition with upshape body instead of italics   
\newtheorem{definition}[theorem]{Definition}

\newtheorem{example}[theorem]{Example}

\deffootnote{1em}{1em}{\textsuperscript{\thefootnotemark\ }}

\title{Automated Complexity Analysis Based on the Dependency Pair Method%
\thanks{This research is partly supported by FWF (Austrian Science Fund) project P20133, the Grant-in-Aid for Young Scientists Nos.~20800022 and~22700009 of
the Japan Society for the Promotion of Science, and Leading Project e-Society (MEXT of Japan), and STARC.}}

\author{Nao Hirokawa\\
School of Information Science,\\
Japan Advanced Institute of Science and Technology, Japan,\\
\email{hirokawa@jaist.ac.jp}
\and Georg Moser\\
Institute of Computer Science,\\
University of Innsbruck, Austria\\
\email{georg.moser@uibk.ac.at}
}

\date{June 2011}

\begin{document}

\maketitle

\begin{abstract}
This article is concerned with automated complexity analysis of
term rewrite systems.
Since these systems underlie much of 
declarative programming, time complexity of functions defined by rewrite systems
is of particular interest. Among other results, 
we present a variant of the dependency pair method for
analysing runtime complexities of term rewrite systems automatically. 
The established results significantly extent previously known techniques: 
we give examples of rewrite systems
subject to our methods that could previously not been analysed
automatically. Furthermore, the techniques have been implemented
in the Tyrolean Complexity Tool. 
We provide ample numerical data for assessing the viability of the method.

\end{abstract}

\newpage
\tableofcontents

\section{Introduction} \label{Introduction}

This article is concerned with automated complexity analysis of
term rewrite systems (TRSs for short).  Since these systems underlie much of 
declarative programming, time complexity of functions defined by TRSs
is of particular interest.

Several notions to assess the complexity of a terminating TRS
have been proposed in the literature, compare~\cite{CKS:1989,HofbauerLautemann:1989,CL:1992,HM:2008}.
The conceptually simplest one was suggested by Hofbauer and Lautemann in~\cite{HofbauerLautemann:1989}: 
the complexity of a given TRS is measured as the maximal length of derivation sequences. 
More precisely, the \emph{derivational complexity function} with respect to a terminating TRS 
relates the maximal derivation height to the size of the initial term.
However, when analysing complexity of a function, it is natural to refine
derivational complexity so that only terms whose arguments are constructor terms
are employed. Conclusively the \emph{runtime complexity function} with respect to a TRS 
relates the length of the longest derivation sequence to the size
of the initial term, where the arguments are supposed 
to be in normal form. This terminology was suggested in~\cite{HM:2008}. 
A related notion has been studied in~\cite{CKS:1989}, 
where it is augmented by an \emph{average case}
analysis. Finally~\cite{CL:1992} studies the complexity of the functions \emph{computed}
by a given TRS. This latter notion is extensively studied within 
\emph{implicit computational complexity theory} (\emph{ICC} for short), see~\cite{BMR:2009} for an overview.
A conceptual difference from runtime complexity is that polynomial 
computability addresses the number of steps by means of (deterministic) 
Turing machines, while runtime complexity measures the number of rewrite 
steps which is closely related to operational semantics of programs.
For instance, a statement like a quadratic complexity of sort 
algorithm is in the latter sense.

This article presents methods for (over-)estimating runtime complexity 
automatically. We establish the following results:
\begin{enumerate}
\item We extend the applicability of direct techniques for complexity results
by showing how the monotonicity constraints can be significantly weakened
through the employ of \emph{usable replacement maps}.

\item We revisit the \emph{dependency pair method} in the context of complexity
analysis. The dependency pair method is originally developed for proving 
termination~\cite{ArtsGiesl:2000}, and known as one of the most 
successful methods in automated termination analysis.  

\item We introduce the \emph{weight gap principle} which allows the 
estimation of the complexity of a TRS in a modular way.

\item We revisit 
the dependency graph analysis 
of the dependency pair method in the context of complexity
analysis.  For that we introduce a suitable notion of \emph{path analysis}
that allows to modularise complexity analysis further.
\end{enumerate}
Note that while we have taken seminal ideas from termination 
analysis as starting points, often the underlying principles 
are crucially different from those used in termination 
analysis.

A preliminary version of this article appeared in~\cite{HM:2008,HM:2008b}. Apart
from the correction of some shortcomings, we extend our earlier work in the
following way:  First, all results on usable replacement maps are new
(see Section~\ref{CSR}). Second, the side condition for the weight gap 
principle~\cite[Theorem~24]{HM:2008} is corrected in
Section~\ref{semantical gap}.  Thirdly, the weight gap principle 
is extended by exploiting the initial term conditions and is 
generalised by means of matrix interpretations (see Section~\ref{semantical gap}). 
Finally, the applicability of the path analysis is strengthened in
comparison to the conference version~\cite{HM:2008b} (see Section~\ref{DG}).

The remainder of this article is organised as follows.
In the next section we recall basic notions.
We define runtime complexity and a subclass of matrix 
interpretations for its analysis
in Section~\ref{Runtime Complexity}.
In Section~\ref{CSR} we relate context-sensitive rewriting to 
runtime complexity.  In the next sections several ingredients in
the dependency pair method are recapitulated for complexity analysis:
dependency pairs and usable rules (Section~\ref{dependency pairs}),
reduction pairs via the weight gap principle (Section~\ref{semantical gap}), 
and dependency graphs (Section~\ref{DG}).
In order to access viability of the presented techniques
all techniques have been implemented in the \emph{Tyrolean Complexity Tool}%
\footnote{\url{http://cl-informatik.uibk.ac.at/software/tct/}.}
(\TCT\ for short) and its empirical data is provided in Section~\ref{Experiments}.
Finally we conclude the article by mentioning related works
in Section~\ref{Conclusion}.

\section{Preliminaries} \label{Preliminaries}

We assume familiarity with term rewriting~\cite{BaaderNipkow:1998,Terese} but 
briefly review basic concepts and notations from term rewriting, 
relative rewriting, and context-sensitive rewriting. 
Moreover, we recall matrix interpretations.

\subsection{Rewriting}

Let $\VS$ denote a countably infinite set of variables and $\FS$ a 
signature,
such that $\FS$ contains at least one constant. 
The set of terms over $\FS$ and $\VS$ is denoted by 
$\TERMS$. The \emph{root symbol} of a term $t$, denoted as $\rt(t)$, 
is either $t$ itself, if $t \in \VS$, or the symbol $f$, if $t = f(\seq{t})$. 
The \emph{set of position} $\Pos(t)$ of a term $t$ is defined as
usual. We write $\Pos_{\GG}(t) \subseteq \Pos(t)$ for the set of
positions of subterms, whose root symbol is contained in $\GG \subseteq \FS$.
The subterm of $t$ at position $p$ is denoted as $\atpos{t}{p}$, 
and $t[u]_p$ denotes the term that is obtained from $t$ by replacing
the subterm at $p$ by $u$.
The subterm relation is denoted as~$\subterm$. 
$\Var(t)$ denotes the set of variables occurring in a term $t$.
The \emph{size} $\size{t}$ of a term is defined 
as the number of symbols in~$t$:
\begin{equation*}
  \size{t} \defsym
  \begin{cases}
    1 & \text{if $t$ is a variable} \tkom\\
    1+ \sum_{1 \leqslant i \leqslant n} \size{t_i} & \text{if $t=f(t_1,\dots,t_n)$} \tpkt
  \end{cases}
\end{equation*}

A \emph{term rewrite system} (\emph{TRS}) $\RS$ over
$\TERMS$ is a \emph{finite} set of rewrite
rules $l \to r$, such that $l \notin \VS$ and $\Var(l) \supseteq \Var(r)$.
The smallest rewrite relation that contains $\RS$ is denoted by
$\to_{\RS}$. 
The transitive closure of $\to_{\RS}$ is denoted by $\rstrew{\RS}$, and
its transitive and reflexive closure by $\rssrew{\RS}$.
We simply write $\to$ for $\to_{\RS}$ if $\RS$ is clear from context.
Let $s$ and $t$ be terms. If exactly $n$ steps are performed to rewrite $s$
to $t$ we write $s \to^n t$.
Sometimes a derivation $s = s_0 \to s_1 \to \cdots \to s_n = t$ is denoted as 
$A \colon s \rss t$ and its length $n$ is referred to as $\card{A}$.
A term $s \in \TERMS$ is called a \emph{normal form} if there is no
$t \in \TERMS$ such that $s \to t$. 
With $\NF(\RS)$ we denote the set of all normal forms of a term rewrite
system $\RS$.
The \emph{innermost rewrite relation} $\irew{\RS}$
of a TRS $\RS$ is defined on terms as follows: $s \irew{\RS} t$ if 
there exist a rewrite rule $l \to r \in \RS$, a context $C$, and
a substitution $\sigma$ such that $s = C[l\sigma]$, $t = C[r\sigma]$,
and all proper subterms of $l\sigma$ are normal forms of $\RS$.  
\emph{Defined symbols} of $\RS$ are symbols appearing at root in left-hand 
sides of $\RS$.
The set of defined function symbols is denoted as $\DS$, while the 
\emph{constructor symbols} $\FF \setminus \DD$ are collected in $\CS$.
We call a term $t = f(\seq{t})$ \emph{basic} or \emph{constructor based} 
if $f \in \DS$ and $t_i \in \TA(\CS,\VS)$
for all $1 \leqslant i \leqslant n$.
The set of all basic terms are denoted by $\TB$.
A TRS $\RS$ is called \emph{duplicating} if there exists a rule 
$l \to r \in \RS$ such that a variable occurs more often in $r$ than in $l$. 
We call a TRS \emph{(innermost) terminating} if no infinite 
(innermost) rewrite sequence exists. 

We recall the notion of \emph{relative rewriting}, 
cf.~\cite{Geser:1990,Terese}.
Let $\RS$ and $\SS$ be TRSs.
The relative TRS $\RS/\SS$ is the pair $(\RS, \SS)$.
We define 
${s \rsrew{\RS / \RSS} t} \defsym 
{s \rssrew{\RSS} \cdot \rsrew{\RS} \cdot \rssrew{\RSS} t}$
and we call $\rsrew{\RS /\RSS}$ the 
\emph{relative rewrite relation} of $\RS$ over $\RSS$.
Note that ${\rsrew{\RS / \RSS}} = {\rsrew{\RS}}$, 
if $\SS = \varnothing$. 
$\RS / \RSS$ is called \emph{terminating} if $\rsrew{\RS / \RSS}$ is well-founded.
In order to generalise the innermost rewriting relation
to relative rewriting, we introduce the slightly technical
construction of the \emph{restricted} rewrite relation,
compare~\cite{T07}.
The \emph{restricted rewrite relation $\toss{\QS}_{\RS}$}
is the restriction of $\rsrew{\RS}$ where 
all arguments of the redex are in
normal form with respect to the TRS $\QS$.
We define the \emph{innermost relative rewriting relation} 
(denoted as $\irew{\RS/\RSS}$) as
follows:
\begin{equation*}
{\irew{\RS/\RSS}} \defsym
  {{\toss{\RS \cup \RSS}_{\RSS}^{\ast}} \cdot
  {\toss{\RS \cup \RSS}_{\RS}} \cdot
  {\toss{\RS \cup \RSS}_{\RSS}^{\ast}}} \tkom
\end{equation*}

We briefly recall context-sensitive rewriting.  A replacement map $\mu$ is
a function with $\mu(f) \subseteq \{1,\ldots, n\}$ for all $n$-ary functions
with $n \geqslant 1$.
The set $\Pos_\mu(t)$ of \emph{$\mu$-replacing positions} in $t$ is defined as follows:
\begin{equation*}
\Pos_\mu(t) =
\begin{cases}
\{ \epsilon \} & \text{if $t$ is a variable} \tkom \\
\{ \epsilon \} \cup \{ ip \mid \text{$i \in \mu(f)$ and 
$p \in \Pos_\mu(t_i)$} \}
& \text{if $t = f(\seq{t})$} \tpkt
\end{cases}
\end{equation*}
A \emph{$\mu$-step} $s \muto{\mu} t$ is a rewrite step $s \to t$ whose rewrite 
position is in $\Pos_\mu(s)$.
The set of all non-$\mu$-replacing positions in $t$ is denoted by
$\NPos_\mu(t)$; namely, $\NPos_\mu(t) \defsym \Pos(t) \setminus \Pos_\mu(t)$.

\subsection{Matrix Interpretations}

One of the most powerful and popular techniques for analysing derivational 
complexities is use of orders induced from matrix interpretations~\cite{EWZ08}.
In order to define it first we define (weakly) monotone algebras.

A \emph{proper order} is a transitive and irreflexive relation and
a \emph{preorder} 
(or \emph{quasi-order})
is a transitive and reflexive relation. 
A proper order $\succ$ is \emph{well-founded} if there is 
no infinite decreasing sequence $t_1 \succ t_2 \succ t_3 \cdots$.
We say a proper order $\succ$ and a TRS $\RS$ are \emph{compatible}
if $\RS \subseteq {\succ}$. 

An $\FS$-\emph{algebra} $\A$ consists of a carrier set $A$ and a collection
of interpretations $f_\A$ for each function symbol in $\FS$. 
By $\eval{\alpha}{\A}(\cdot)$ we denote the usual evaluation function
of $\A$ according to an assignment $\alpha$ which maps variables 
to values in~$A$.
A \emph{monotone $\FS$-algebra} is a pair $(\A,\succ)$
where $\A$ is an $\FS$-algebra and $\succ$ is a proper
order such that for every function symbol $f\in\FS$, $f_\A$
is strictly monotone in all coordinates with respect to $\succ$.
A \emph{weakly monotone $\FS$-algebra} $(\A,\succcurlyeq)$ 
is defined similarly, but for every function symbol $f\in\FS$, 
it suffices that $f_\A$ is weakly monotone in all coordinates 
(with respect to the quasi-order $\succcurlyeq$). 
A monotone $\FS$-algebra $(\A,\succ)$ is called \emph{well-founded} 
if $\succ$ is well-founded. We write \emph{WMA} instead of
well-founded monotone algebra. 

Any (weakly) monotone $\FS$-algebra $(\A,\R)$ induces a binary relation
$\R_\A$ on terms: define $s \R_\A t$ if 
$\eval{\alpha}{\A}(s) \R \eval{\alpha}{\A}(t)$ for all assignments $\alpha$. 
Clearly if $\R$ is a proper order (quasi-order), then $\R_\A$ is a proper order 
(quasi-order) on terms and if $\R$ is a well-founded, then 
$\R_\A$ is well-founded on terms. 
We say $\A$ is \emph{compatible} with a TRS $\RS$ if 
${\RS} \subseteq {\R_\A}$.  
Let $\geqord{\A}$ denote the quasi-order induced by a weakly monotone
algebra $(\A,\succcurlyeq)$, then $\eqord{\A}$ denotes the equivalence (on
terms) induced by $\geqord{\A}$. 
Let $\mu$ denote a replacement map. Then we call a well-founded algebra
$(\A,\succ)$ \emph{$\mu$-monotone} if for every function symbol $f \in \FS$, $f_\A$ is
strictly monotone 
\emph{on} $\mu(f)$, i.e., $f_\A$ is strictly monotone with respect to
every argument position in $\mu(f)$.
Similarly a (strict) relation $\R$ is called 
$\mu$-monotone if (strictly) monotone on $\mu(f)$ for all  $f \in \FS$.
Let $\RS$ be a TRS compatible with a $\mu$-monotone
relation $\R$. Then clearly any $\mu$-step $s \muto{\mu} t$ implies
$s \R t$.

We recall the concept of \emph{matrix interpretations} on natural numbers
(see~\cite{EWZ08} but compare also~\cite{HW06}).
Let $\FS$ denote a signature. 
We fix a dimension $d\in\N$ and use the set $\N^d$ as the carrier of 
an algebra $\A$, together with the following extension of the
natural order $>$ on $\N$:
\begin{equation*}
(x_1,x_2,\ldots,x_d) > (y_1,y_2,\ldots,y_d) \defeqv
x_1>y_1 \wedge x_2 \geqslant y_2 \wedge \ldots \wedge x_d \geqslant y_d \tpkt
\end{equation*}
Let $\mu$ be a replacement map.
For each $n$-ary function symbol $f$, we choose as an interpretation
a linear function of the following shape:
\begin{equation*}
f_{\A} \colon (\vec{v}_1,\ldots,\vec{v}_n) 
\mapsto F_1 \vec{v}_1 + \cdots + F_n \vec{v}_n + \vec{f}
\tkom
\end{equation*}
where $\vec{v}_1,\ldots,\vec{v}_n$ are (column) vectors of variables,
$F_1,\ldots,F_n$ are matrices (each of size $d \times d$), and $\vec{f}$ is a vector over $\N$.
Moreover, suppose for any $i \in \mu(f)$
the top left entry $(F_i)_{1,1}$ is positive. Then it is easy to see that the
algebra $\A$ forms a $\mu$-monotone WMA. 
Let $\A$ be a matrix interpretation, let $\alpha_0$ denotes the assignment mapping 
any variable to $\vec{0}$, i.e., $\alpha_0(x) = \vec{0}$ for all $x \in \VS$, and let $t$ be a term. 
In the following we write $[t]$, $[t]_j$
as an abbreviation for $\eval{\alpha_0}{\A}(t)$, or
$\left( \eval{\alpha_0}{\A}(t) \right)_j$ ($1 \leqslant j \leqslant d$), 
respectively, if the algebra $\A$ is clear from the context.

\section{Runtime Complexity} \label{Runtime Complexity}

In this section we formalise runtime complexity and then define a subclass of matrix 
interpretations that give polynomial upper-bounds.

The \emph{derivation height} of a term $s$ with respect to a 
well-founded, finitely branching relation $\to$ is defined as: 
$\dheight(s,\to) = \max\{ n \mid \exists t \; s \to^n t \}$. 
Let $\RS$ be a TRS and $T$ be a set of terms. 
The \emph{complexity function with respect to a relation $\to$ on $T$}
is defined as follows:
\begin{equation*}
\comp(n, T, \rew) = \max\{ \dheight(t, \rew) \mid 
\text{$t \in T$ and $\size{t} \leqslant n$}\} \tpkt
\end{equation*}
In particular we are interested in the 
(innermost) complexity with respect to $\rsrew{\RS}$ ($\irew{\RS}$)
on the set $\TB$ of all \emph{basic} terms.%

\begin{definition}
Let $\RS$ be a TRS.
We define the \emph{runtime complexity function} $\Rc{\RS}(n)$,
the \emph{innermost runtime complexity function} $\Rci{\RS}(n)$, and 
the \emph{derivational complexity function} $\Dc{\RS}(n)$ as
$\comp(n, {\TB}, \rsrew{\RS})$,
$\comp(n, {\TB}, \irew{\RS})$, and
$\comp(n, \TA(\FS,\VS), \rsrew{\RS})$, respectively.
\end{definition}

Note that the above complexity functions need not be defined, as the
rewrite relation $\rsrew{\RS}$ is not always well-founded \emph{and} finitely branching.
We sometimes say the (innermost) runtime complexity of $\RS$ is \emph{linear},
\emph{quadratic}, or \emph{polynomial} if there exists a (linear, quadratic)
polynomial $p(n)$ such that $\Rcpareni{\RS}(n) \leqslant p(n)$ for sufficiently
large $n$. The (innermost) runtime complexity of $\RS$ is called \emph{exponential} if
there exist constants $c$, $d$ with $c,d \geqslant 2$ 
such that $c^n \leqslant \Rcpareni{\RS}(n) \leqslant d^n$ for sufficiently
large $n$. 

The next example illustrates a difference between 
derivational complexity and runtime complexity.
\begin{example} \label{ex:1}
\label{ex:div}
Consider the following TRS $\RSdiv$%
\footnote{This is Example~3.1 in Arts and Giesl's collection of TRSs~\cite{ArtsGiesl:2001}.}
\begin{alignat*}{4}
1\colon &\; &
x - \m{0} & \to x
&
\qquad
3\colon &\; &
\m{0}  \div \m{s}(y) & \to \m{0}
\\
2\colon &\; &
\m{s}(x) - \m{s}(y) & \to x - y
&
\qquad
4\colon &\; &
\m{s}(x)  \div \m{s}(y) & \to \m{s}((x - y) \div \m{s}(y))
\tpkt
\end{alignat*}
Although the functions \emph{computed} by $\RSdiv$ are obviously feasible this
is not reflected in the derivational complexity of $\RSdiv$. 
Consider rule 4, which we abbreviate as $C[x] \to D[x,x]$. 
Since the maximal derivation height starting with $C^n[x]$ equals $2^{n-1}$ 
for all $n > 0$, $\RSdiv$ admits (at least) exponential derivational complexity.
In general any duplicating TRS admits (at least)
exponential derivational complexity.
\end{example}

In general it is not possible to bound $\Dc{\RS}$ polynomially in $\Rc{\RS}$, 
as witnessed by Example~\ref{ex:1} and the observation that
the runtime complexity of $\RS$ is linear
(see~Example~\ref{ex:1:ua}, below). We will use Example~\ref{ex:1}
as our running example.

Below we define classes of orders whose compatibility with a TRS $\RS$ 
bounds its runtime complexity from the above.
Note that $\dheight(t, {\succ})$ is undefined, if the relation $\succ$
is not well-founded or not finitely branching.
In fact compatibility of a constructor TRS with the 
polynomial path order $>_{\m{pop*}}$ (\cite{AM:2009}) induces 
polynomial innermost runtime complexity, whereas 
$\m{f}(x) >_{\m{pop*}} \cdots >_{\m{pop*}} \cdots >_{\m{pop*}} 
\m{g}^2(x) >_{\m{pop*}} \m{g}(x) >_{\m{pop*}} x$ holds when
precedence $\m{f} > \m{g}$ is used. 
Hence $\dheight(t, {>_{\m{pop*}}})$ is undefined, while the order
$>_{\m{pop*}}$ can be employed in complexity analysis.

\begin{definition} \label{d:collapsible}
Let $\R$ be a binary relation over terms, let $\succ$ be a proper order on terms, 
and let $\Slow$ denote a mapping associating a term with a natural number.
Then $\succ$ is \emph{$\Slow$-collapsible on $\R$} if
$\Slow(s) > \Slow(t)$, whenever ${s} \R {t}$ and ${s} \succ {t}$ holds. 
An order $\succ$ is \emph{collapsible (on $\R$)},
if there is a mapping $\Slow$ such that $\succ$ is $\Slow$-collapsible
(on $\R$).
\end{definition}

\begin{lemma}
Let $\R$ be a finitely branching and well-founded relation.
Further, let $\succ$ be a $\Slow$-collapsible order with ${\R} \subseteq {\succ}$.
Then $\dheight(t,{\R}) \leqslant \Slow(t)$ holds for all terms $t$.
\end{lemma}

The alert reader will have noticed that any proper order $\succ$ is 
collapsible on a finitely branching and well-founded relation $\R$: 
simply set $\Slow(t) \defsym \dheight(t,{\R})$. 
However, this observation is of limited use if we wish to bound the derivation
height of $t$ in independence of $\R$.

If a TRS $\RS$ and a $\mu$-monotone matrix interpretation $\A$ are compatible,
$\Slow(t)$ can be given by $[t]_1$.
In order to estimate derivational or runtime complexity,
one needs to associate $[t]_1$ to  $|t|$.
For this sake we define degrees of matrix interpretations.

\begin{definition}
A matrix interpretation is of \emph{(basic) degree} $d$ if
there is a constant $c$ such that $[t]_i \leqslant c \cdot |t|^d$ 
for all (basic) terms $t$ and $i$, respectively.
\end{definition}

An \emph{upper triangular complexity matrix} is a matrix $M$ in $\N^{d\times d}$
such that we have $M_{j,k}=0$ for all $1 \leqslant k < j \leqslant d$, 
and $M_{j,j}\leqslant 1$ for all $1 \leqslant j \leqslant d$.
We say that a WMA $\A$ is a \emph{triangular matrix interpretation} 
(\emph{TMI} for short) if $\A$ is a matrix interpretation (over $\N$) and all
matrices employed are of upper triangular complexity form. 
It is easy to define triangular matrix interpretations, 
such that an algebra $\A$ based on such an interpretation, 
forms a well-founded \emph{weakly} monotone algebra. To simplify
notation we will also refer to $\A$ as a TMI, if no confusion can arise
from this.
A TMI $\A$ of dimension 1, that is a linear polynomial, is called a
\emph{strongly linear interpretation} (\emph{SLI} for short) if all
interpretation functions $f_{\A}$ are strongly linear.
Here a polynomial $P(x_1,\dots,x_n)$ is strong linear if 
$P(x_1,\dots,x_n) = x_1 + \cdots + x_n + c$.

\begin{lemma} \label{l:8}
Let $\A$ be a TMI and let $M$ denote the component-wise
maximum of all matrices occurring in $\A$. 
Further, let $d$ denote the number of ones occurring along the diagonal of $M$. 
Then for all $1 \leqslant i,j \leqslant d$ we have $(M^n)_{i,j} = \bO(n^{d-1})$.
\end{lemma}
\begin{proof}
The lemma is a direct consequence of Lemma~4 in~\cite{NZM:2010} together with 
the observation that for any triangular complexity
matrix, the diagonal entries denote the multiset
of eigenvalues. 
\end{proof}

\begin{lemma} \label{l:9}
Let $\A$ and $d$ be defined as in Lemma~\ref{l:8}. 
Then $\A$ is of degree $d$.
\end{lemma}
\begin{proof}
For any (triangular) matrix interpretation $\A$, 
there exist vectors $\vec{v}_i$ and a vector $\vec{w}$ 
such that the evaluation $[t]$ of $t$ can be written as follows:
\begin{equation*} 
  [t] = \sum_{i=1}^{\ell} \vec{v}_i + \vec{w} \tkom
\end{equation*}
where each vector $\vec{v}_i$ is the product of those matrices employed
in the interpretation of function symbols in $\A$ and a vector representing
the constant part of a function interpretation. 
It is not difficult to see that there is a one-to-one correspondence between
the number of vectors $\vec{v}_1,\dots,\vec{v}_{\ell}$ and
the number of subterms of $t$ and thus $\ell = \size{t}$. Moreover for each
$\vec{v}_i$ the number of products is less than the depth of $t$ and thus
bounded by $\size{t}$. 
In addition, due to Lemma~\ref{l:8} the entries of the vectors 
$\vec{v}_i$ and $\vec{w}$ are bounded by a polynomial of degree at most $d-1$. 
Thus for all $1 \leqslant j \leqslant d$, there exists $k \leqslant d$ such 
that $([t])_j = \bO(\size{t}^{k})$.
\end{proof}

\begin{theorem}{\cite[Theorem~9]{NZM:2010},\cite{W:2010}}
\label{t:TMI}
Let $\A$ and $d$ be defined as in Lemma~\ref{l:8}. 
Then, $\gord{\A}$ is $\bO(n^d)$-collapsible.
\end{theorem}
\begin{proof}
The theorem is a direct consequence of Lemmas~\ref{l:8} and~\ref{l:9}.  
\end{proof}

In order to cope with runtime complexity, a similar idea to restricted 
polynomial interpretations (see \cite{BCMT:2001}) can be integrated to
triangle matrix interpretations. We call $\A$ a 
\emph{restricted matrix interpretation} (\emph{RMI} for short) if 
$\A$ is a matrix interpretation, but for each constructor symbol $f \in \FS$,
the interpretation $f_\A$ of $f$ employs upper triangular complexity
matrices, only.
The next theorem is a direct consequence of the definitions
in conjunction with Lemma~\ref{l:9}.

\begin{theorem} \label{t:rmi}
Let $\A$ be an RMI and let $t$ be a basic term. 
Further, let $M$ denote the component-wise maximum of all matrices 
used for the interpretation of constructor symbol, 
and let $d$ denote the number of ones occurring along the diagonal of $M$.
Then $\A$ is of basic degree $d$.  Furthermore, if $M$ is the unit matrix
then $\A$ is of basic degree $1$.
\end{theorem}

\section{Usable Replacement Maps}
\label{CSR}

Unfortunately, there is no RMI compatible with
the TRS of our running example.  The reason is that the 
monotonicity requirement of TMI is too severe for complexity analysis. 
Inspired by the idea of Fern\'andez \cite{F:2005}, we show how 
context-sensitive rewriting is used in complexity analysis.
Here we briefly explain our idea.
Let $\mathbf{n}$ denote the numeral $s^n(\m{0})$.
Consider the derivation from $\mathbf{4} \div \mathbf{2}$:
\begin{equation*}
\underline{\mathbf{4} \div \mathbf{2}}
\to \m{s}(\underline{(\mathbf{3} - \mathbf{1})} \div \mathbf{2}) 
\to \m{s}((\underline{\mathbf{2} - \m{0}}) \div \mathbf{2})
\to \m{s}(\underline{\mathbf{2} \div \mathbf{2}}) 
\to \cdots
\end{equation*}
where redexes are underlined.
Observe that e.g. any second argument of $\div$ is never rewritten.  
More precisely, any derivation from a basic term consists of
only $\mu$-steps with the replacement map $\mu$: 
$\mu(\m{s}) = \mu({\div}) = \{1\}$ and 
$\mu({-}) = \varnothing$. 

We present a simple method 
based on a variant of $\ICAP$ in \cite{GTS05}
to estimate a suitable replacement map.
Let $\mu$ be a replacement map. Clearly the function $\mu$ is representable as
set of ordered pairs $(f,i)$. Below we often confuse the notation of
$\mu$ as a function or as a set.
Recall that $\Pos_\mu(t)$ denotes the set of \emph{$\mu$-replacing positions} 
in $t$ and $\NPos_\mu(t) = \Pos(t) \setminus \Pos_\mu(t)$. 
Further, a term $t$ is a \emph{$\mu$-replacing term} 
with respect to a TRS $\RS$ if ${\atpos{t}{p}} \not\in {\NF(\RS)}$
implies that $p \in Pos_\mu(t)$. The set of all $\mu$-replacing terms is denoted by $\MUTERM{\mu}$.
In the following $\RS$ will always denote a TRS.

\begin{definition}
\label{d:uargs}
Let $\RS$ be a TRS and let $\mu$ be a replacement map.
We defined the operator $\Upsilon^\RS$ as follows:
\begin{equation*}
\Upsilon^\RS(\mu) \defsym \{ (f,i) \mid 
\text{$l \to C[f(\seq{r})] \in \RS$ and $\MUCAP{\mu}{l}{r_i} \neq r_i$} \} \tpkt
\end{equation*}
Here $\MUCAP{\mu}{s}{t}$ is inductively defined on $t$ as follows:
\begin{equation*}
\MUCAP{\mu}{s}{t} =
\begin{cases}
t & \text{$t = \atpos{s}{p}$ for some $p \in \NPos_\mu(s)$} \tkom\\
u & \text{if $t = f(\seq{t})$ and 
$u$ and $l$ unify for no
$l \to r \in \RS$} \tkom \\
y & \text{otherwise} \tkom
\end{cases}
\end{equation*}
where, $u = f(\MUCAP{\mu}{s}{t_1},\ldots,\MUCAP{\mu}{s}{t_n})$,
$y$ is a fresh variable, and $\Var(l) \cap \Var(u) = \varnothing$
is assumed.  
\end{definition}

We define the \emph{innermost usable replacement map} $\IURM{\RS}$
as follows $\IURM{\RS} \defsym \Upsilon^\RS(\varnothing)$ and let
the \emph{usable replacement map} $\URM{\RS}$ denote the 
least fixed point of $\Upsilon^\RS$.  The existence of $\Upsilon^\RS$ follows from 
the monotonicity of $\Upsilon^\RS$.
If $\RS$ is clear from context, we simple write
$\imu$, $\tmu$, and $\Upsilon$, respectively.
Usable replacement maps satisfy a desired property for runtime
complexity analysis.  In order to see it several preliminary 
lemmas are necessary.  

First we take a look at 
$\MUCAP{\mu}{s}{t}$.
Suppose $s \in \TT(\mu)$: observe that the function $\MUCAP{\mu}{s}{t}$ replaces 
a subterm $u$ of $t$ by a fresh variable if $u\sigma$ is a redex %, 
for some $s\sigma \in \TT(\mu)$.
This is exemplified below.
\begin{example}
Consider the TRS $\RSdiv$. Let $l \to r$ be rule 4, namely,
$l = \m{s}(x) \div \m{s}(y)$ and $r = \m{s}((x - y) \div \m{s}(y))$.
Suppose $\mu(f) = \varnothing$ for all functions $f$
and let $w$ and $z$ be fresh variables.
The next table summarises $\MUCAP{\mu}{l}{t}$ for each proper subterm 
$t$ in $r$.  To see the computation process, we also indicate 
the term $u$ in Definition~\ref{d:uargs}.
\begin{center}
\begin{tabular}{@{}c@{\qquad}ccccc@{}}
\hline
$t$ & $x$ & $y$ & $x - y$ & $\m{s}(y)$ & $(x-y) \div \m{s}(y)$ \\
$u$ & --  & --  & $x - y$ & $\m{s}(y)$ & $w \div \m{s}(y)$ \\
$\MUCAP{\mu}{l}{t}$ & $x$ & $y$ & $w$ & $\m{s}(y)$ & $z$ \\
\hline
\end{tabular}
\end{center}
By underlining proper subterms $t$ in $r$ such that 
$\MUCAP{\mu}{l}{t} \neq t$, we have 
\[
\m{s}(\underline{\underline{(x-y)} \div \m{s}(y)})
\]
which indicates $(\m{s}, 1), ({\div},1) \in \Upsilon(\mu)$.
\end{example}

The next lemma states a role of $\MUCAP{\mu}{s}{t}$.
\begin{lemma}
\label{l:MUCAP}
If $s\sigma \in \MUTERM{\mu}$ and $\MUCAP{\mu}{s}{t} = t$ then
$t\sigma \in \NF(\RS)$.
\end{lemma}
\begin{proof}
We use induction on $t$.
Suppose $s\sigma \in \MUTERM{\mu}$ and $\MUCAP{\mu}{s}{t} = t$.
If $t = \atpos{s}{p}$ for some $p \in \NPos_\mu(s)$ then 
$t\sigma = \atpos{(s\sigma)}{p} \in \NF$ follows by definition of
$\MUTERM{\mu}$.

We can assume that $t = f(\seq{t})$. Assume
otherwise that $t = x \in \VS$, then $\MUCAP{\mu}{s}{x} = x$ entails that
$x\sigma$ occurs at a non-$\mu$-replacing position in $s\sigma$. Hence
$x\sigma \in \NF$ follows from $s\sigma \in \MUTERM{\mu}$.
Moreover, by assumption we have:
\begin{enumerate}
\item \label{en:MUCAP:1}
$\MUCAP{\mu}{s}{t_i} = t_i$ for each $i$, and
\item \label{en:MUCAP:2}
there is no rule  $l \to r \in \RS$ such that $t$ and $l$ unify.  
\end{enumerate}
Due to~\ref{en:MUCAP:2}) $l\sigma$ is not reducible at the root, and
the induction hypothesis yields $t_i\sigma \in \NF$ 
because of~\ref{en:MUCAP:1}).
Therefore, we obtain $t\sigma \in \NF$.
\end{proof}

For a smooth inductive proof of our key lemma we prepare a 
characterisation of the set of $\mu$-replacing terms $\MUTERM{\mu}$.
\begin{definition}
\label{d:upsilon}
The set
$\{ (f,i) \mid \text{$f(\seq{t}) \subterm t$ and $t_i \not\in \NF(\RS)$} \}$
is denoted by $\upsilon(t)$.
\end{definition}
%
%The next lemma shows that the set of $\mu$-replacing terms $\MUTERM{\mu}$
%can be characterised through Definition~\ref{d:upsilon}.
%
\begin{lemma}
\label{l:MUTERM}
$\MUTERM{\mu} = \{ t \mid \upsilon(t) \subseteq \mu \}$.
\end{lemma}
\begin{proof}
The inclusion from left to right essentially follows from the
definitions. Let $t \in \MUTERM{\mu}$ and let $(f,i) \in \upsilon(t)$.
We show $(f,i) \in \mu$.
By Definition~\ref{d:upsilon} there is a position $p \in \Pos(t)$ with
$\atpos{t}{p} = f(\seq{t})$ and ${\atpos{t}{pi}} \not\in {\NF}$.  
Thus $pi \in \Pos_\mu(t)$ and $i \in \Pos_\mu(\atpos{t}{p})$. 
Hence $(f,i) \in \mu$ is concluded.

Next we consider the reverse direction 
${\{ t \mid \upsilon(t) \subseteq \mu \}} \subseteq {\MUTERM{\mu}}$. 
Let $t$ be a minimal term such that $\upsilon(t) \subseteq \mu$ and
$t \not\in \MUTERM{\mu}$. 
One can write $t = f(\seq{t})$.  Then,
there exists a position $p \in \NPos_\mu(t)$ such that $\atpos{t}{p} \not\in \NF$.
Because $\epsilon \not\in \NPos_\mu(t)$ holds in general,
$p$ is of the form $iq$ with $i \in \NN$.
As $iq \in \NPos_\mu(t)$ 
one of $(f, i) \not\in \mu$ or $q \in \NPos_\mu(\atpos{t}{i})$ must hold.
As $t$ is minimal and ${\atpos{t}{iq}} \not\in {\NF}$ implies that 
${\atpos{t}{i}} \not\in {\NF}$, we have $(f, i) \not\in \mu$. However,
by Definition~\ref{d:upsilon}, $(f,i) \in \upsilon(t) \subseteq \mu$.
Contradiction.
\end{proof}

The next lemma about the operator $\Upsilon$ is a key for the main theorem.  Note 
that every subterm of a $\mu$-replacing term is a $\mu$-replacing term.

\begin{lemma}
\label{l:Upsilon}
If $l \to r \in \RS$ and $l\sigma \in \MUTERM{\mu}$ then
$r\sigma \in \MUTERM{\mu \cup \Upsilon(\mu)}$.
\end{lemma}
\begin{proof}
Let $l \to r \in \RS$ and suppose $l\sigma \in \MUTERM{\mu}$.
By Lemma~\ref{l:MUTERM} we have
\begin{equation*}
  \MUTERM{\mu} = \{ t \mid \upsilon(t) \subseteq \mu \} \qquad
  \MUTERM{\mu \cup \Upsilon(\mu)} = \{ t \mid {\upsilon(t)} \subseteq {\mu \cup \Upsilon(\mu)}\} \tpkt
\end{equation*}
Hence it is sufficient to show 
$\upsilon(r\sigma) \subseteq \mu \cup \Upsilon(\mu)$.
Let $(f,i) \in \upsilon(r\sigma)$.  There is $p \in \Pos(r\sigma)$
with ${\atpos{r\sigma}{p}} = {f(\seq{t})}$ and $t_i \not\in \NF$.
If $p$ is below some variable position of $r$,
${\atpos{r\sigma}{p}}$ is a subterm of $l\sigma$, and thus
$\upsilon(\atpos{r\sigma}{p}) \subseteq \upsilon(l\sigma) \subseteq \mu$.
Otherwise, $p$ is a non-variable position of~$r$.
We may write $\atpos{r}{p} = f(\seq{r})$ and $r_i\sigma = t_i \not\in \NF$.
Due to Lemma~\ref{l:MUCAP} we obtain $\MUCAP{\mu}{l}{r_i} \neq r_i$. 
Therefore, $(f,i) \in \Upsilon(\mu)$.
\end{proof}

Remark that if $s, t \in \MUTERM{\mu}$ and $p \in \Pos_\mu(s)$ then
$s[t]_p \in \MUTERM{\mu}$.

\begin{lemma}
\label{l:mu-closed}
The following implications hold.
\begin{enumerate}
\item \label{en::mu-closed:1}
If $s \in \MUTERM{\imu}$ and $s \ito t$ then
$t \in \MUTERM{\imu}$.
\item \label{en::mu-closed:2}
If $s \in \MUTERM{\tmu}$ and  $s \to t$ then $t \in \MUTERM{\tmu}$.
\end{enumerate}
\end{lemma}
\begin{proof}
We show property~\ref{en::mu-closed:1}).
Suppose $s \in \MUTERM{\imu}$ and $s \ito t$ is a rewrite step at $p$.
Due to the definition of innermost rewriting, we have
$\atpos{s}{p} \in \MUTERM{\varnothing}$.  Hence, $\atpos{t}{p} \in \MUTERM{\imu}$ 
is obtained by Lemma~\ref{l:Upsilon}.
Because $s\in \MUTERM{\imu}$ we have $p \in \Pos_\imu(s)$. 
Hence due to $\atpos{t}{p} \in \MUTERM{\imu}$ we conclude $t = s[\atpos{t}{p}]_p \in \MUTERM{\imu}$
due to the above remark.
The proof of~\ref{en::mu-closed:2}) proceeds along the same pattern and is left to the reader.
\end{proof}

We arrive at the main result of this section.

\begin{theorem}
\label{t:mu-inclusion}
Let $\RS$ be a TRS, and let $\Desc(L)$ denote the descendants of
the set of terms $L$. Then $\rsisrew{\RS}(\MUTERM{\varnothing}) \subseteq \MUTERM{\imu}$
and $\rssrew{\RS}(\MUTERM{\varnothing}) \subseteq \MUTERM{\tmu}$.
\end{theorem}
\begin{proof}
Recall that $\Desc(L) \defsym \{t \mid \text{$\exists s \in L$ such that $s \to^\ast t$}\}$. We focus on the second part of the theorem, where we have to
prove that $t \in \MUTERM{\tmu}$, whenever there exists $s \in \MUTERM{\varnothing}$
such that $s \rssrew{\RS} t$. 
As $\MUTERM{\varnothing} \subseteq \MUTERM{\tmu}$
this follows directly from Lemma~\ref{l:mu-closed}.
\end{proof}

Note that $\MUTERM{\varnothing}$ is the set of all argument normalised 
terms.  Therefore, ${\TB} \subseteq {\MUTERM{\varnothing}}$.
The following corollary to Theorem~\ref{t:mu-inclusion} is immediate.
\begin{corollary}
Let $\RS$ be a 
TRS and let $\muto{\imu}$, $\muto{\tmu}$ denote 
the $\imu$-step and $\tmu$-step relation, respectively.  
Then for all terminating terms $t \in \TB$ we have 
$\dheight(t, {\rsirew{\RS}}) \leqslant \dheight(t, {\muto{\imu}})$
and $\dheight(t, {\rsrew{\RS}}) \leqslant \dheight(t, {\muto{\tmu}})$.
\end{corollary}

An advantage of the use of context-sensitive rewriting is that the
compatibility requirement of monotone algebra in termination or complexity
analysis is relaxed to $\mu$-monotone 
algebra.
We illustrate its use in the next example.
\begin{example}
\label{ex:1:ua}
Recall the TRS $\RSdiv$ given in Example~\ref{ex:div} above.
The usable argument positions are as follows:
\begin{equation*}
\imu(\m{-}) = \varnothing \quad \imu(\ms) = \imu(\m{\div}) = \{1\} \qquad
\tmu(\ms) = \tmu(\m{-}) = \tmu(\m{\div}) = \{1\} \tpkt
\end{equation*}
Consider the $1$-dimensional RMI $\A$ (i.e., linear polynomial interpretations)
with
\begin{xalignat*}{4}
\m{0}_\A & = 1
& 
\m{s}_\A(x) & = x + 2
&
{-_\A}(x, y) & = x + 1
&
{\div_\A}(x, y) & = 3x
\tpkt
\end{xalignat*}
which is strictly $\imu$-monotone and $\tmu$-monotone.
The rules in $\RSdiv$ are interpreted and ordered as follows.
\begin{alignat*}{6}
& 1\colon\quad &  x + 1 & > x
& \qquad
& 3\colon\quad &      3 & > 1
\\
& 2\colon\quad &  x + 3 & > x + 2
& \qquad
& 4\colon\quad & 3x + 6 & > 3x + 5
\tpkt
\end{alignat*}
Therefore, $\RSdiv \subseteq {>_\A}$ holds.  
By an application of Theorem~\ref{t:rmi} we conclude that the (innermost) 
runtime complexity is \emph{linear}, which is optimal.
\end{example}

We cast the observations in the example into another corollary to
Theorem~\ref{t:mu-inclusion}.

\begin{corollary}
\label{c:mu-inclusion}
Let $\RS$ be a TRS and let $\A$ be a $d$-degree $\imu$-monotone 
(or $\tmu$-monotone) RMI compatible with $\RS$. Then the (innermost) runtime complexity
function $\Rcpareni{\RS}$ with respect to $\RS$ is 
bounded by a $d$-degree polynomial.
\end{corollary}
\begin{proof}
It suffices to consider the case for full rewriting. Let $s$, $t$
be terms such that $s \rsrew{\RS} t$. By the theorem, we have
$s \muto{\tmu} t$. Furthermore, by assumption 
${\RS} \subseteq {\gord{\A}}$ and for any $f \in \FS$, 
$f_\A$ is strictly monotone on all $\tmu(f)$. Thus $s \gord{\A} t$ follows. 
Finally, the corollary follows by application of Theorem~\ref{t:rmi}.
\end{proof}

We link Theorem~\ref{t:mu-inclusion} to related work by Fern\'andez ~\cite{F:2005}.
In~\cite{F:2005} it is shown how context-sensitive rewriting
is used for proving innermost termination.  

\begin{proposition}[\textnormal{\cite{F:2005}}]
\label{c:sin}
A TRS $\RS$ is innermost terminating if $\muto{\imu}$ is terminating.
\end{proposition}
\begin{proof}
We show the contraposition.
If $\RS$ is not innermost terminating, there is an infinite sequence
$t_0 \ito t_1 \ito t_2 \ito \cdots$, where $t_0 \in \MUTERM{\varnothing}$.
From Theorem~\ref{t:mu-inclusion} and Lemma~\ref{l:mu-closed} 
we obtain $t_0 \muto{\imu} t_1 \muto{\imu} t_2 \muto{\imu} \cdots$.  
Hence, $\muto{\imu}$ is not terminating.
\end{proof}

One might think that a similar claim holds for full termination
if one uses $\tmu$. The next examples clarifies that this
is not the case.
\begin{example}
Consider the famous Toyama's example $\RS$
\begin{xalignat*}{3}
\m{f}(\m{a},\m{b},x) & \to \m{f}(x,x,x) &
\m{g}(x,y) & \to x &
\m{g}(x,y) & \to y 
\tpkt
\end{xalignat*}
The replacement map $\tmu$ is empty.  Thus, the algebra $\A$ over $\NN$
\begin{xalignat*}{4}
\m{f}_\A(x,y,z) & = \max\{x - y, 0\} &
\m{g}_\A(x,y) & = x + y + 1 &
\m{a}_\A & = 1 &
\m{b}_\A & = 0 
\tpkt
\end{xalignat*}
is $\tmu$-monotone and we have $\RS \subseteq {>_\A}$.  However,
we should not conclude termination of $\RS$, 
because $\m{f}(\m{a},\m{b},\m{g}(\m{a},\m{b}))$ is non-terminating.
\end{example}

\section{Weak Dependency Pairs}
\label{dependency pairs}

In Section~\ref{CSR} we investigated argument positions of rewrite steps.
This section is concerned about contexts surrounding rewrite steps.
Recall the derivation: 
\label{eq:5}
\begin{alignat*}{3}
\boxed{\mathbf{4} \div \mathbf{2}}
&  ~\rsrew{\RSdiv}  \m{s}(\,\boxed{(\mathbf{3} - \mathbf{1}) \div \mathbf{2}}\,) 
&& ~\rsnrew{\RSdiv}{2} \m{s}(\,\boxed{\mathbf{2} \div \mathbf{2}}\,)
\\
&  ~\rsrew{\RSdiv}   \m{s}(\m{s}(\,\boxed{(\mathbf{1} - \mathbf{1}) \div \mathbf{2}}\,))
&& ~\rsnrew{\RSdiv}{2} \m{s}(\m{s}(\,\boxed{\m{0} \div \mathbf{2}}\,))
\\
&  ~\rsrew{\RSdiv}   \m{s}(\m{s}(\m{0})) \tkom
\end{alignat*}
where we boxed outermost occurrences of defined symbols.  
Obviously, their surrounding contexts are not rewritten.
Here an idea is to simulate rewrite steps from basic terms
with new rewrite rules, obtained by dropping unnecessary contexts.
In termination analysis this method is known as the dependency pair 
method~\cite{ArtsGiesl:2000}.  We recast its main
ingredient called dependency pairs.

Let $X$ be a set of symbols.
We write $\SC{t_1, \ldots, t_n}_X$ to denote $C[t_1,\ldots,t_n]$,
whenever $\rt(t_i) \in X$ for all $1 \leqslant i \leqslant n$
and $C$ is an $n$-hole context containing no $X$-symbols.  
(Note that the context $C$ may be degenerate and doesn't contain a hole $\ctx$
or it may be that $C$ is a hole.)
Then, every term $t$ can be uniquely written in
the form $\SC{t_1, \ldots, t_n}_X$.

\begin{lemma} \label{l:1}
Let $t$ be a terminating term, and let $\sigma$ be a substitution.
Then $\dheight(t\sigma, \to_\RS) = \sum_{1 \leqslant i \leqslant n} \dheight(t_i\sigma, \rsrew{\RS})$,
whenever $t = \SC{t_1, \ldots, t_n}_{\DD \cup \VV}$.
\end{lemma}
The idea is to replace such a $n$-hole context with a fresh $n$-ary function 
symbol.
We define the function $\COM$ as a mapping from tuples of terms to terms 
as follows: $\COM(\seq{t})$ is $t_1$ if $n = 1$, and
$c(t_1,\ldots,t_n)$ otherwise.
Here $c$ is a fresh $n$-ary function symbol called \emph{compound symbol}.
The above lemma motivates the next definition of 
\emph{weak dependency pairs}. 

\begin{definition} \label{d:WDP}
Let $t$ be a term. We set $t^\sharp \defsym t$ if $t \in \VV$, and 
$t^\sharp \defsym f^\sharp(t_1,\dots,t_n)$ if $t = f(\seq{t})$.
Here $f^\sharp$ is a new $n$-ary function symbol called 
\emph{dependency pair symbol}. For a signature
$\FS$, we define $\FS^\sharp = \FS \cup \{f^\sharp \mid f\in \FS\}$.
Let $\RS$ be a TRS. 
If $l \rew r \in \RS$ and $r = \SC{\seq{u}}_{\DD \cup \VV}$ then 
the rewrite rule $l^\sharp \to \COM(u_1^\sharp,\ldots,u_n^\sharp)$ is called
a \emph{weak dependency pair} of $\RS$.
The set of all weak dependency pairs is denoted by $\WDP(\RS)$.
\end{definition}
While dependency pair symbols are defined with respect to $\WDP(\RS)$,
these symbols are not defined with respect to the original system $\RS$. In the
sequel defined symbols refer to the defined function symbols of $\RS$.

\begin{example}[continued from Example~\ref{ex:1}] \label{ex:1:WDP}
The set $\WDP(\RSdiv)$ consists of the next four weak dependency pairs:
\begin{alignat*}{4}
5\colon &\;& x -^\sharp \m{0} & \to x & \qquad
7\colon &\;& \m{0}  \div^\sharp \m{s}(y) &\to \m{c} \\
6\colon && \m{s}(x) -^\sharp \m{s}(y) & \to x -^\sharp y & \qquad
8\colon && \m{s}(x) \div^\sharp \m{s}(y) &\to (x - y) \div^\sharp \m{s}(y)
\tpkt
\end{alignat*}
Here $\m{c}$ denotes a fresh compound symbols of arity $0$.
\end{example}

The derivation on page~\pageref{eq:5} corresponds to the derivation of 
$\WDP(\RSdiv) \cup \RSdiv$:%
\label{eq:7}%
\begin{alignat*}{3}
\mathbf{4} \div^\sharp \mathbf{2} 
&  ~\rsrew{\WDP(\RSdiv)}~ (\mathbf{3} - \mathbf{1}) \div^\sharp \mathbf{2}
&& ~\rsnrew{\RSdiv}{2}  \mathbf{2} \div^\sharp \mathbf{2}
\\
&  ~\rsrew{\WDP(\RSdiv)}~ (\mathbf{1} - \mathbf{1}) \div^\sharp \mathbf{2}
&& ~\rsnrew{\RSdiv}{2}       \m{0} \div^\sharp \mathbf{2}
\\
&  ~\rsrew{\WDP(\RSdiv)}~ \m{c}
\tkom
\end{alignat*}
which preserves the length.  The next lemma states that this is
generally true.
                
\begin{lemma} \label{l:2}
Let $t \in \TA(\FS,\VS)$ be a terminating term with defined root.
Then we obtain:
$\dheight(t,\rsrew{\RS})=\dheight(t^\sharp, \rsrew{\WDP(\RS) \cup \RS})$.
\end{lemma}
\begin{proof}
We show 
$\dheight(t, \rsrew{\RS}) \leqslant
 \dheight(t^\sharp, \rsrew{\WDP(\RS) \cup \RS})$ by 
induction on $\dheight(t, {\rsrew{\RS}})$.
Let $\ell = \dheight(t, {\rsrew{\RS}})$.
If $\ell = 0$, the inequality is trivial.
Suppose $\ell > 0$.  Then there exists a term $u$ such that
$t \rsrew{\RS} u$ and $\dheight(u, \rsrew{\RS}) = \ell - 1$.
We distinguish two cases depending on the rewrite position $p$.
\begin{enumerate}
\item
If $p$ is a position below the root, then clearly $\rt(u) = \rt(t) \in \DD$ and
$t^\sharp \rsrew{\RS} u^\sharp$.  
%The induction hypothesis yields
Induction hypothesis yields
$\dheight(u, {\rsrew{\RS}}) \leqslant
 \dheight(u^\sharp, {\rsrew{\WDP(\RS) \cup \RS}})$,
and we obtain
$\ell \leqslant \dheight(t^\sharp, \rsrew{\WDP(\RS) \cup \RS})$.
\item
If $p$ is a root position, then there exist a rewrite rule $l \to r \in \RS$ and
a substitution $\sigma$ such that $t = l\sigma$ and $u = r\sigma$.
There exists a context $C$ such that $r = \SC{\seq{u}}_{\DD \cup \VV}$ and thus by definition
$l^\sharp \to \COM(u_1^\sharp,\ldots,u_n^\sharp) \in \WDP(\RS)$ such that
$t^\sharp = l^\sharp\sigma$. Now, either $u_i \in \VV$ or $\rt(u_i) \in \DD$ 
for every $1 \leqslant i \leqslant n$.
Suppose $u_i \in \VV$. Then $u_i^\sharp\sigma = u_i\sigma$ and clearly no
dependency pair symbol can occur and thus,
\begin{equation*}
\dheight(u_i\sigma, \rsrew{\RS})
= \dheight(u_i^\sharp\sigma, \rsrew{\RS})
= \dheight(u_i^\sharp\sigma, \rsrew{\WDP(\RS) \cup \RS}) \tpkt
\end{equation*}
Otherwise, if $\rt(u_i) \in \DD$ then $u_i^\sharp\sigma = (u_i\sigma)^\sharp$.
Hence $\dheight(u_i\sigma, \rsrew{\RS}) \leqslant \dheight(u, \rsrew{\RS}) < \ell$,
and we conclude
$\dheight(u_i\sigma, \rsrew{\RS}) \leqslant
 \dheight(u_i^\sharp\sigma, \rsrew{\WDP(\RS) \cup \RS})$
from the induction hypothesis. Therefore,
\begin{align*}
\ell 
& = \dheight(u, \rsrew{\RS}) + 1
\\
& = \sum_{1 \leqslant i \leqslant n} \dheight(u_i\sigma, \rsrew{\RS}) + 1
  \leqslant \sum_{1 \leqslant i \leqslant n} \dheight(u_i^\sharp\sigma, \rsrew{\WDP(\RS) \cup \RS}) + 1\\
& = \dheight(\COM(u_1^\sharp,\ldots,u_n^\sharp)\sigma, 
                     \rsrew{\WDP(\RS) \cup \RS}) + 1
 \leqslant \dheight(t^\sharp, \rsrew{\WDP(\RS) \cup \RS}) \tpkt
\end{align*}
Here we used Lemma~\ref{l:1} for the second equality.
\end{enumerate}
Note that $t$ is $\RS$-reducible if and only if $t^\sharp$ is
$\WDP(\RS)\cup\RS$-reducible. Hence as $t$ is terminating, 
$t^\sharp$ is terminating on $\rsrew{\WDP(\RS)\cup\RS}$. Thus, similarly, 
$\dheight(t, \rsrew{\RS}) \geqslant
 \dheight(t^\sharp, \rsrew{\WDP(\RS) \cup \RS})$ 
is shown by induction on $\dheight(t^{\sharp}, \rsrew{\WDP(\RS)\cup\RS})$.
\end{proof}

In the case of innermost rewriting we need not include collapsing dependency
pairs as in Definition~\ref{d:WDP}. This is guaranteed by the next lemma.
\begin{lemma} 
Let $t$ be a terminating term and $\sigma$ a substitution such that $x\sigma$ is a normal form of $\RS$ 
for all $x \in \Var(t)$.
Then $\dheight(t\sigma, \rsrew{\RS}) =  \sum_{1 \leqslant i \leqslant n} \dheight(t_i\sigma, \rsrew{\RS})$,
whenever $t = \SC{t_1, \ldots, t_n}_\DD$.
\end{lemma}
\begin{definition}
Let $\RS$ be a TRS. 
If $l \rew r \in \RS$ and $r = \SC{\seq{u}}_\DD$ then 
the rewrite rule $l^\sharp \to \COM(u_1^\sharp,\ldots,u_n^\sharp)$
is called a \emph{weak innermost dependency pair} of $\RS$.
The set of all weak innermost dependency pairs is denoted by $\WIDP(\RS)$.
\end{definition}

\begin{example}[continued from Example~\ref{ex:1}] \label{ex:1:WIDP}
The set $\WIDP(\RSdiv)$ consists of the next three weak innermost dependency pairs
(with respect to $\ito$):
\begin{alignat*}{4}
&\;& \m{s}(x) -^\sharp \m{s}(y) & \to x -^\sharp y & 
&\;& \m{0}  \div^\sharp \m{s}(y) & \to \m{c} \\
&& \m{s}(x) \div^\sharp \m{s}(y) &\to (x - y) \div^\sharp \m{s}(y) 
\tpkt
& \qquad &&&
\end{alignat*}
\end{example}

The next lemma adapts Lemma~\ref{l:2} to innermost rewriting.
\begin{lemma} \label{l:3}
Let $t$ be an innermost terminating term in $\TA(\FS,\VS)$ with $\rt(t) \in \DD$.
We have $\dheight(t, \irew{\RS}) = \dheight(t^\sharp, \irew{\WIDP(\RS) \cup \RS})$.
\end{lemma}

Looking at the simulated version of the derivation on page~\pageref{eq:5},
%\begin{alignat*}{2}
%\mathbf{4} \div^\sharp \mathbf{2}
%&  \to_{\WDP(\RS)} (\mathbf{3} - \mathbf{1}) \div^\sharp \mathbf{2}
%&& \to_{\{1,2\}}^* \mathbf{2} \div^\sharp \mathbf{2} 
%\\
%&  \to_{\WDP(\RS)} (\mathbf{1} - \mathbf{1}) \div^\sharp \mathbf{2}
%&& \to_{\{1,2\}}^* \m{0} \div^\sharp \mathbf{2} 
%\\
%& \to_{\WDP(\RS)} \m{c}
%\end{alignat*}
rules 1 and 2 are used, but neither rule 3 nor 4 is used in the $\RS$-steps.
In general we can approximate a subsystem of a TRS that can be used in 
derivations from basic terms, by employing 
the notion of usable rules in the dependency pair method 
(cf.~\cite{ArtsGiesl:2000,GTSF06,HirokawaMiddeldorp:2007}).

\begin{definition}
We write ${f} \depends {g}$ if there exists a rewrite rule
$l \to r \in \RS$ such that $f = \rt(l)$ and $g$ is a defined function
symbol in $\Fun(r)$. 
For a set $\GG$ of defined function symbols we
denote by $\RS{\restriction}\GG$ the set of
rewrite rules $l \to r \in \RS$ with $\rt(l) \in \GG$. 
The set $\UU(t)$ of \emph{usable rules} of a term $t$ is defined as
$\RS{\restriction}\{ g \mid \text{${f} \depends^* {g}$ for some $f \in \Fun(t)$} \}$.
Finally, if $\PP$ is a set of (weak) dependency pairs
then $\UU(\PP) = \bigcup_{l \to r \in \PP} \UU(r)$.
\end{definition}

\begin{example}[continued from Examples~\ref{ex:1:WDP} and~\ref{ex:1:WIDP}]
\label{ex:1:usable}
The set $\UU(\WDP(\RSdiv))$ of usable rules for the weak dependency pairs 
consists of the two rules:
\begin{alignat*}{4}
1\colon &\;&  x - \m{0} & \to x & \qquad
2\colon &\;&  \m{s}(x) - \m{s}(y) &\to x -y 
\tpkt
\end{alignat*}
Note that we have that $\UU(\WDP(\RSdiv)) = \UU(\WIDP(\RSdiv))$.
\end{example}

We show a usable rule criterion for complexity analysis 
by exploiting the property that the starting terms are basic.
Recall that $\TB$ denotes the set of basic terms;
we set $\TBS = \{t^{\sharp} \mid t \in \TB \}$.

\begin{lemma} \label{l:5}
Let $\PP$ be a set of weak dependency pairs and let 
$(t_i)_{i = 0, 1, \ldots}$ be a (finite or infinite) derivation of
$\PP \cup \RS$. If $t_0 \in \TBS$ then $(t_i)_{i = 0, 1, \ldots}$ is 
a derivation of $\PP \cup \UU(\PP)$.
\end{lemma}
\begin{proof}
Let $\GG$ be the set of all non-usable symbols with respect to $\PP$.
We write $P(t)$ if ${\atpos{t}{q}} \in {\NF(\RS)}$ for all $q \in \Pos_\GG(t)$. 
First we prove by induction on $i$ that $P(t_i)$ holds for all $i$.
\begin{enumerate}
\item 
Assume $i = 0$.  Since $t_0 \in \TBS$, we have $t_0 \in \NF(\RS)$
and thus ${\atpos{t}{p}} \in {\NF(\RS)}$ for all positions $p$. 
The assertion $P$ follows trivially.
\item
Suppose $i > 0$.  
By induction hypothesis, $P(t_{i-1})$ holds, i.e., 
there exist $p \in \Pos(t_{i-1})$, a substitution $\sigma$, and
$l \rew r \in \UU(\PP) \cup \PP$, such that
${\atpos{t_{i-1}}{p}} = l\sigma$ and $\atpos{t_i}{p} = r\sigma$.
In order to show property $P$ for $t_i$, we fix a position $q \in \Pos_\GG(t)$. We
have to show $\atpos{t_i}{q} \in \NF(\RS)$. We distinguish three subcases:
\begin{itemize}
\item
Suppose that $q$ is above $p$. Then $\atpos{t_{i-1}}{q}$ is 
reducible, but this contradicts the induction hypothesis $P(t_{i-1})$.
\item
Suppose $p$ and $q$ are parallel but distinct.
Since $\atpos{t_{i-1}}{q} = \atpos{t_i}{q} \in \NF(\RS)$ holds,
we obtain $P(t_i)$.
\item
Otherwise, $q$ is below $p$. Then, $\atpos{t_i}{q}$ is a subterm of $r\sigma$.  
Because $r$ contains no $\GG$-symbols by the definition of usable symbols,
$\atpos{t_i}{q}$ is a subterm of $x\sigma$ for some 
$x \in \Var(r) \subseteq \Var(l)$.
Therefore, $\atpos{t_i}{q}$ is also a subterm of $\atpos{t_{i-1}}{q}$, 
from which $\atpos{t_i}{q} \in \NF(\RS)$ follows.  We obtain $P(t_i)$.
\end{itemize}
\end{enumerate}
Hence property $P$ holds for all $t_i$ in the assumed derivation. Thus
any reduction step $t_i \rsrew{\RS \cup \PP} t_{i+1}$ can be simulated
by a step $t_i \rsrew{\UU(\PP) \cup \PP} t_{i+1}$. From this the
lemma follows.
\end{proof}

Note that the proof technique adopted for termination 
analysis~\cite{GTSF06,HirokawaMiddeldorp:2007} cannot be 
directly used in this context.
The technique transforms terms in a derivation to exclude non-usable rules.  However, since the
size of the initial term increases, this technique does not 
suit to our use. 
On the other hand, the transformation employed in~\cite{HirokawaMiddeldorp:2007}
is adaptable to a complexity analysis in the large, cf.~\cite{MS:2010}.
% GM: as a reminder
% Example 7.1. in MS:2010 shows that the standard proof of usable rules in
% the innermost case is not applicable to derivational complexity.

The next theorem follows from Lemmas~\ref{l:2} and~\ref{l:3}
in conjunction with the above Lemma~\ref{l:5}.
It adapts the usable rule criteria to complexity analysis.

\begin{theorem} \label{t:dp:usable}
Let $\RS$ be a TRS and let $t \in \TB$. 
If $t$ is terminating with respect to $\rew$ then
$\dheight(t, \rew) = \dheight(t^{\sharp},\rsrew{\PP \cup \UU(\PP)})$, 
where $\rew$ denotes $\rsrew{\RS}$ or $\irew{\RS}$
depending on whether $\PP = \WDP(\RS)$ or $\PP = \WIDP(\RS)$.
\end{theorem}

To clarify the applicability  of the theorem in complexity analysis, we 
instantiate the theorem by considering RMIs.

\begin{corollary} \label{c:dp:usable}
Let $\RS$ be a TRS, let $\mu$ be the (innermost) usable replacement map 
and let $\PP = \WDP(\RS)$ 
(or $\PP = \WIDP(\RS)$). 
If $\PP \cup \UU(\PP)$ is compatible with a 
$d$-degree $\mu$-monotone RMI $\A$, then
the (innermost) runtime complexity function $\rc^{(\m{i})}_{\RS}$ 
with respect to $\RS$ is bounded by a $d$-degree polynomial.
\end{corollary}
\begin{proof}
For simplicity we suppose $\PP = \WDP(\RS)$ and let $\A$ be
a $\mu$-monotone RMI of degree $d$.
Compatibility of $\A$ with $\PP \cup \UU(\PP)$ 
implies the well-foundedness of the relation $\rsrew{\PP \cup \UU(\PP)}$ on
the set of terms $\TBS$, cf.~Theorem~\ref{t:mu-inclusion}. This
in turn implies the well-foundedness of $\rsrew{\RS}$, cf.~Lemma~\ref{l:5}.
Hence Theorem~\ref{t:dp:usable} is applicable and we 
conclude $\dheight(t,\rsrew{\RS}) = \dheight(t^\sharp, \rsrew{\PP \cup \UU(\PP)})$.
On the other hand, due to Theorem~\ref{t:rmi} compatibility with $\A$ implies that
$\dheight(t^\sharp, \rsrew{\PP \cup \UU(\PP)}) = \bO(\size{t^\sharp}^d)$.
As $\size{t^\sharp} = \size{t}$, we can combine these equalities 
to conclude polynomial runtime complexity of $\RS$.
\end{proof}

The below given example applies Corollary~\ref{c:dp:usable} to the
motivating Example~\ref{ex:1} introduced in Section~\ref{Introduction}.

\begin{example}[continued from Example~\ref{ex:1:usable}]
\label{ex:3}
Consider the TRS $\RSdiv$ for division used as running
example; the weak dependency pairs $\PP \defsym \WDP(\RSdiv)$ 
are given in Example~\ref{ex:1:WDP}. 
We have $\UU(\PP) = \{ 1, 2 \}$ and let $\SS = \PP \cup \UU(\PP)$. 
The usable replacement map $\mu \defsym \URM{\SS}$ is defined 
as follows:
\begin{alignat*}{3}
\mu(\m{s}) & = \mu(\m{-}) = \mu(\m{-}^\sharp) = \varnothing 
& \qquad &
\mu(\div^\sharp) & = \{1\} 
\tpkt
\end{alignat*}
Note that $\URM{\SS}$ is smaller than $\URM{\RS}$ on $\FF$
(see Example~\ref{ex:1:ua}).
Consider the $1$-dimensional RMI $\A$ with
$\m{0}_\A = \m{c}_\A = \m{d}_\A = 0$,
$\m{s}_\A(x) = x + 2$,
$\m{-}_\A(x, y) = \m{-}^\sharp_\A(x, y) = x + 1$,
and $\div^\sharp_\A(x, y) = x + 1$.
The algebra $\A$ is strictly monotone on all usable argument positions
and the rules in $\SS$ are interpreted and ordered as follows:
\begin{alignat*}{9}
& 1\colon\quad & x + 1 & > x
& \qquad
& 5\colon\quad & 1 & > 0
& \qquad
& 7\colon\quad & 1 & > 0
\\
& 2\colon\quad & x + 3 & > x + 1
& \qquad
& 6\colon\quad & x + 3 & > x + 1
& \qquad
& 8\colon\quad & x + 3 & > x + 2 
\tpkt
\end{alignat*}
Therefore, $\SS$ is compatible with $\A$ and the
runtime complexity function $\Rc{\RS}$ is linear.
Remark that by looking at the coefficients of the interpretations 
more precise bound can be inferred.  Since all coefficients 
are at most one, we obtain $\Rc{\RS}(n) \leqslant n + c$ for some $c \in \NN$.
\end{example}

It is worth stressing that it is (often) easier to analyse the 
complexity of $\PP \cup \UU(\PP)$  than the complexity of 
$\RS$.  This is exemplified by the next example.

\begin{example}
\label{ex:8}
Consider the TRS $\RSdiff$
\begin{align*}
\m{D}(\m{c}) & \to \m{0}
& 
\m{D}(x + y) & \to \m{D}(x) + \m{D}(y)
&
\m{D}(x \times y) & \to 
(y \times \m{D}(x)) + (x \times \m{D}(y))
\\
\m{D}(\m{t}) & \to \m{1}
&
\m{D}(x - y) & \to \m{D}(x) - \m{D}(y)
\tpkt
\end{align*}
There is no $1$-dimensional $\tmu$-monotone RMI compatible with $\RSdiff$.
On the other hand  
$\WDP(\RSdiff)$ consists of the five pairs
\begin{align*}
\m{D}^\sharp(\m{c}) & \to \m{c_1}
& 
\m{D}^\sharp(x + y) & \to \m{c_3}(\m{D}^\sharp(x), \m{D}^\sharp(y))
&
\m{D}^\sharp(x \times y) & \to 
\m{c_5}(y, \m{D}^\sharp(x), x, \m{D}^\sharp(y))
\\
\m{D}^\sharp(\m{t}) & \to \m{c_2}
&
\m{D}^\sharp(x - y) & \to \m{c_4}(\m{D}^\sharp(x), \m{D}^\sharp(y))
\tkom
\end{align*}
and $\UU(\WDP(\RSdiff)) = \varnothing$.
The usable replacement map $\tmu$ for $\WDP(\RSdiff) \cup \UU(\RSdiff)$ is
defined as $\tmu(\m{c_3}) = \tmu(\m{c_4}) = \{1,2\}$, 
$\tmu(\m{c_5}) = \{2,4\}$, and $\tmu(f) = \varnothing$ for all other
symbols $f$.  Since the $1$-dimensional $\tmu$-monotone RMI $\A$
with
\begin{align*}
& 
\m{D}^\sharp_\A(x) = 2x 
\qquad \m{c}_\A = \m{t}_\A = 1 
\qquad {+}_\A(x,y) = {-}_\A(x,y) = {\times}_\A(x,y) = x + y + 1
\\
&
\m{c_1}_\A = \m{c_2}_\A = 0 
\qquad \m{c_3}_\A(x,y) = \m{c_4}_\A(x,y) = x + y
\qquad \m{c_5}_\A(x,y,z,w) = y + w
\tkom
\end{align*}
is compatible with $\RSdiff$, linear runtime complexity of $\RSdiff$ is concluded. 
Remark that this bound is optimal.
\end{example}

We conclude this section by discussing the (in-)applicability 
of standard dependency pairs (see~\cite{ArtsGiesl:2000}) in complexity analysis.
For that we recall the definition of standard dependency pairs.

\begin{definition}[\cite{ArtsGiesl:2000}]
\label{d:DP}
The set $\DP(\RS)$ of (standard) \emph{dependency pairs} of a TRS $\RS$ 
is defined as 
$\{ l^\sharp \to u^{\sharp} \mid l \to r \in \RS, 
\text{$u \subterm r$, $\rt(u)$ is defined, and $u \not\prsubterm l$} \}$.
\end{definition}

The next example shows that Lemma~\ref{l:2} (Lemma~\ref{l:3})
does not hold if we replace weak (innermost) dependency pairs with 
standard dependency pairs.
\begin{example}
\label{ex:6} 
Consider the one-rule TRS $\RS$:
$\m{f}(\m{s}(x)) \to \m{g}(\m{f}(x), \m{f}(x))$.
$\DP(\RS)$ is the singleton of $\m{f}^\sharp(\m{s}(x)) \to \m{f}^\sharp(x)$.
Let $t_n = \m{f}(\m{s}^n(x))$ for each $n \geqslant 0$.
Since $t_{n+1} \rsrew{\RS} \m{g}(t_n,t_n)$ holds for
all $n \geqslant 0$, it is easy to see 
$\dheight(t_{n+1}, \rsrew{\RS}) \geqslant 2^n$, while
$\dheight(t_{n+1}^\sharp, \rsrew{\DP(\RS) \cup \RS}) = n$.
\end{example}

\section{The Weight Gap Principle}
\label{semantical gap}

Let $\PP = \WDP(\RSdiv)$ and recall the derivation over
$\PP \cup \RSdiv$ on page~\pageref{eq:7}. 
This derivation
can be represented as derivation of $\PP$ modulo $\UU(\PP)$:
\label{eq:relative}
\begin{equation*}
\mathbf{4} \div^\sharp \mathbf{2} 
~\rsrew{\PP/\UU(\PP)}~ \mathbf{2} \div^\sharp \mathbf{2}
~\rsrew{\PP/\UU(\PP)}~ \m{0} \div^\sharp \mathbf{2}
~\rsrew{\PP/\UU(\PP)}~ \m{c} \tpkt
\end{equation*}
As we see later linear runtime complexity of $\UU(\PP)$ and $\PP/\UU(\PP)$
can be easily obtained.  If linear runtime complexity of $\PP \cup \UU(\PP)$
would follow from them, linear runtime complexity of $\RS$ could be 
established in a modular way. 

In order to bound complexity of relative TRSs we define a variant of
a reduction pair~\cite{ArtsGiesl:2000}. 
Note that $\Slow$ is associated to a given collapsible order.

\begin{definition}
A $\mu$-\emph{complexity pair} for a relative TRS $\RS/\RSS$
is a pair $({\gtrsim},{\succ})$ such that $\gtrsim$ is a 
$\mu$-monotone proper order and $\succ$ is a strict order. Moreover
${\gtrsim}$ and ${\succ}$ are compatible, that is, 
${\gtrsim \cdot \succ} \subseteq {\succ}$ or 
${\succ \cdot \gtrsim} \subseteq {\succ}$. 
Finally $\succ$ is collapsible on $\rsrew{\RS/\RSS}$ and
all compound symbols are $\mu$-monotone with respect to $\succ$.
\end{definition}

\begin{lemma}
Let $\PP = \WDP(\RS)$ and $({\gtrsim},{\succ})$ a 
$\tmu^{\PP \cup \UU(\PP)}$-complexity pair for $\PP/\UU(\PP)$.
If $\PP \subseteq {\succ}$ and $\UU(\PP) \subseteq {\gtrsim}$ then
$\dheight(t,\rsrew{\PP/\UU(\PP)}) \leqslant \Slow(t)$ for any $t \in \TBS$.
\end{lemma}

\begin{example}[continued from Example~\ref{ex:3}]
\label{ex:relative}
Consider the $1$-dimensional RMI $\A$ with
\begin{xalignat*}{3}
\m{0}_\A & = \m{c}_\A = \m{d}_\A = 0 
& 
\m{s}_\A(x) & = x + 1
&
{-}_\A(x, y) & = {-}^\sharp_\A(x, y) = \div^\sharp_\A(x, y) = x 
\tkom
\end{xalignat*}
which yields the complexity pair $({\geqord[\geqslant]{\A}},{\gord[>]{\A}})$ for 
$\PP/\UU(\PP)$.  Since ${\PP} \subseteq {\gord[>]{\A}}$ and 
${\UU(\PP)} \subseteq {\geqord[\geqslant]{\A}}$ hold, 
$\comp(n, \TBS, \rsrew{\PP/\UU(\PP)}) = \bO(n)$.
\end{example}

First we show the main theorem of this section.
\begin{definition}
Let $\A$ be a matrix interpretation and let $\RS/\SS$ be a relative TRS. 
A \emph{weight gap} on a set $T$ of terms is a number $\Delta \in \NN$ 
such that $s \in {\to^*_{\RS \cup \SS}}(T)$ and 
$s \to_\RS t$ implies $[t]_1 - [s]_1 \leqslant \Delta$.
\end{definition}

Let $T$ be a set of terms and let $\RS/\SS$ be a relative TRS.

\begin{theorem}
\label{t:wgp}
If $\RS/\SS$ is terminating,
$\A$ admits a \emph{weight gap} $\Delta$ on $T$,
and $\A$ is a matrix interpretation of degree $d$ such that
$\SS$ is compatible with $\A$, then there exists $c \in \NN$ such that 
\(
\dheight(t,{\to_{\RS \cup \SS}}) \leqslant 
(1 + \Delta) \cdot \dheight(t,{\to_{\RS/\SS}}) + c \cdot |t|^d
\)
for all $t \in T$.
Consequently,
\(
\comp(n,T,{\rsrew{\RS \cup \SS}}) =
\bO(\comp(n,T,{\rsrew{\RS/\SS}}) + n^d)
\)
holds.
\end{theorem}
\begin{proof}
Let $m = \dheight(s, {\rsrew{\RS/\SS}})$ and $n = \size{s}$.
Any derivation of $\rsrew{\RS \cup \RSS}$ is representable as follows:
\begin{equation*} 
s = s_0 \to_{\RSS}^{k_0} 
t_0 \to_\RS 
s_1 \to_\RSS^{k_1} 
t_1 \to_\RS \cdots \to_\RSS^{k_m}
t_m \tpkt
\end{equation*}
Without loss of generality we may assume that the derivation is
maximal and ground. We observe:
\begin{enumerate}
\item \label{en:relative:i} 
$k_i \leqslant [s_i]_1 - [t_i]_1$ holds for all $0 \leqslant i \leqslant m$. This is 
because $[s]_1 > [t]_1$, whenever $s \rsrew{\RSS} t$
by the assumption $\SS$ is compatible with $\A$. By definition of $>$, 
we conclude $[s]_1 \geqslant [t]_1 + 1$ whenever $s \rsrew{\RSS} t$. From the
fact that $s_i \to_{\RSS}^{k_i} t_i$ we thus obtain $k_i \leqslant [s_i]_1 - [t_i]_1$.

\item \label{en:relative:ii} 
$([s_{i+1}])_1 \leqslant ([t_i])_1 + \Delta$ holds for all $0 \leqslant i < m$ 
by the assumption.

\item \label{en:relative:iii} 
There exists a number $c$ such that for any term $s \in T$,
$[s]_1 \leqslant c \cdot \size{s}^d$. 
This follows by the degree of $\A$.
\end{enumerate}
We obtain the following inequalities:
\begin{align*}
  \dheight(s_0, \rsrew{\RS \cup \SS}) & = m + k_0 + \dots + k_m \\
  & \leqslant m + ([s_0]_1 - [t_0]_1) + \dots + ([s_m]_1 - [t_m]_1) \\
  & = m + [s_0]_1 + ([s_1]_1 - [t_0]_1) + \dots + ([s_m]_1 - [t_{m-1}]_1) - [t_m]_1 \\
  & \leqslant m + [s_0]_1 + ([t_0]_1 + \Delta - [t_0]_1) + \dots  
   - [t_m]_1\\
  & \leqslant m + [s_0]_1 + m\Delta  - [t_m]_1\\
  & \leqslant m + [s_0]_1 + m\Delta\\
  & \leqslant (1+\Delta) m  + c \cdot \size{s_0}^d \tpkt
\end{align*}
Here we use property~\ref{en:relative:i}) $m$-times in the second line. 
We used property~\ref{en:relative:ii}) in the third line
and property~\ref{en:relative:iii}) in the last line.
\end{proof}

A question is when a weight gap is admitted.  We present two
conditions.  We start with a simple version for derivational
complexity, and then we adapt it for runtime complexity.

We employ a very restrictive form of TMIs.
Every $f \in \FS$ is interpreted by the following restricted linear function:
\begin{equation*}
  f_{\A} \colon (\vec{v}_1,\ldots,\vec{v}_n) 
\mapsto \mathbf{1} \vec{v}_1 + \ldots + \mathbf{1} \vec{v}_n + \vec{f}
\tpkt
\end{equation*}
I.e., the only matrix employed in this interpretation is the unit matrix
$\mathbf{1}$. Such a matrix interpretation is called 
\emph{strongly linear} (\emph{SLMI} for short). 

\begin{lemma}
\label{l:weightgap:i}
If $\RS$ is non-duplicating and $\A$ is an SLMI, then 
$\RS/\SS$ and $\A$ admit a weight gap on all terms.
\end{lemma}
\begin{proof}
Let $\Delta \defsym \max \{ [r]_1 \modminus [l]_1 \mid l \to r \in \RS \}$.
We show that $\Delta$ gives a weight gap.
In proof, we first show the following equality.
\begin{equation} \label{eq:1}
  \Delta = \max \{(\eval{\alpha}{\A}(r))_1 \modminus (\eval{\alpha}{\A}(l))_1 \mid
  l \to r \in \RS, \alpha \colon \VS \to \A \} \tpkt
\end{equation}
Although the proof is not difficult, we give the full account in order
to utilise it later. 
Observe that for any matrix interpretation 
$\A$ and rule ${l \to r} \in {\RS}$, there 
exist matrices (over $\N$) 
$L_1,\dots,L_k$, $R_1,\dots,R_k$ and vectors $\vec{l}$, $\vec{r}$ such that:
\begin{equation*}
 \eval{\alpha}{\A}(l) = \sum_{i=1}^k  L_i \cdot \alpha(x_i) + \vec{l} \hspace{10ex}
 \eval{\alpha}{\A}(r) = \sum_{i=1}^k  R_i \cdot \alpha(x_i) + \vec{r} \tkom
\end{equation*}
where $k$ denotes the cardinality of $\Var(l) \supseteq \Var(r)$.
Conclusively, we obtain:
\begin{equation} \label{eq:2}
\eval{\alpha}{\A}(r) \modminus \eval{\alpha}{\A}(l) = 
  \sum_{i=1}^k (R_i \modminus L_i) \alpha(x_i) + (\vec{r} \modminus \vec{l}) 
  \tpkt
\end{equation}
Here $\modminus$ denotes the natural component-wise 
extension of the modified minus to vectors. 

As $\A$ is an SLMI the matrices $L_i$, $R_i$ are obtained by multiplying or adding
unit matrices, where the latter case can only happen if (at least one)
of the variables $x_i$ occurs multiple times in $l$ or $r$. Due to the
fact that $l \to r$ is non-duplicating, this effect is canceled out.
Thus the right-hand side of~\eqref{eq:2} is independent on the
assignment $\alpha$ and we conclude:
\begin{equation*}
[r]_1 \modminus [l]_1 = (\eval{\alpha}{\A}(r) \modminus \eval{\alpha}{\A}(l))_1 = (\vec{r} \modminus \vec{l})_1 \tpkt
\end{equation*}
By definition $\Delta = \max\{[r]_1 \modminus [l]_1 \mid l \to r \in \RS \}$ and
thus~\eqref{eq:1} follows.

Let $C[\ctx]$ denote a (possible empty) context such that
$s = C[l\sigma] \rsrew{\RS} C[r\sigma] = t$, where ${l \rew r} \in {\RS}$ and
$\sigma$ a substitution. 
We prove the lemma by induction on~$C$. 
\begin{enumerate}
\item \label{en:weightgap:i:i}
Suppose $C[\ctx] = \ctx$, that is, $s = l\sigma$ and $t = r\sigma$. There exists
an assignment $\alpha_1$ such that
$[l\sigma] = \eval{\alpha_1}{\A}(l)$ and
$[r\sigma] = \eval{\alpha_1}{\A}(r)$.
By~\eqref{eq:1} we conclude for the assignment $\alpha_1$:
$(\eval{\alpha_1}{\A}(l))_1 + \Delta \geqslant (\eval{\alpha_1}{\A}(r))_1$.
Therefore in sum we obtain $[s]_1 + \Delta \geqslant [t]_1$.

\item \label{en:weightgap:i:ii}
Suppose $C[\ctx] = f(t_1,\dots,t_{i-1},C'[\ctx],t_{i+1},\dots,t_n)$.
Hence, we obtain:
\begin{align*}
& [f(t_1,\dots,C'[l\sigma],\dots,t_n)]_1 + \Delta  \\
= {} 
& [t_1]_1 + \dots + ([C'[l\sigma]]_1 + \Delta) + \dots + 
  [t_n]_1 + (\vec{f})_1 \\
\geqslant {}
& [t_1]_1 + \dots + [C'[r\sigma]]_1 + \dots + 
  [t_n]_1 + (\vec{f})_1 \\
= {} 
& [f(t_1,\dots,C'[r\sigma],\dots,t_n)]_1 \tkom
\end{align*}
for some vector $\vec{f} \in \N^d$. 
In the first and last line, we employ the fact that $\A$ is strongly
linear. In the second line the induction hypothesis is applied 
together with the (trivial) fact that $\A$ is strictly monotone 
on all arguments of $f$ by definition.
\end{enumerate}
\end{proof}

Note that the combination of Theorem~\ref{t:wgp} and 
Lemma~\ref{l:weightgap:i} corresponds to (the corrected version of) 
Theorem~24 in~\cite{HM:2008}. In~\cite{HM:2008} 1-dimensional SLMIs
are called \emph{strongly linear interpretations} (\emph{SLIs} for short).

\begin{example}
Consider the TRS $\RS$
\begin{align*}
1\colon~ \m{f}(\m{s}(x)) & \to \m{f}(x - \m{s}(\m{0})) &
2\colon~ x - \m{0} & \to x &
3\colon~ \m{s}(x) - \m{s}(y) & \to x - y
\tpkt
\end{align*}
$\PP \defsym \WDP(\RS)$ consists of  the three pairs
\begin{align*}
\m{f}^\sharp(\m{s}(x)) & \to \m{f}^\sharp(x - \m{s}(\m{0})) &
x -^\sharp \m{0} & \to x &
\m{s}(x) -^\sharp \m{s}(y) & \to x -^\sharp y
\tkom
\end{align*}
and $\UU(\PP) = \{ 2,3 \}$.  Obviously $\PP$ is non-duplicating and
there exists an SLI $\A$ with $\UU(\PP) \subseteq {\gord{\A}}$.
Thus, Lemma~\ref{l:weightgap:i} yields a weight gap for $\PP/\UU(\PP)$.
By taking the $1$-dimensional RMI $\BB$ with
\begin{xalignat*}{3}
\m{s}_\BB(x) & = x + 1 &
{-}_\BB(x,y) & = x &
\m{f}_\BB(x) & = \m{f}^\sharp_\BB(x) = x 
\\
\m{0}_\BB & = 0 &
{-^\sharp}_\BB(x,y) & = x + 1 
\tkom
\end{xalignat*}
we obtain $\PP \subseteq {\gord{\BB}}$ and 
$\UU(\PP) \subseteq {\geqord{\BB}}$.  
Therefore, $\comp(n, \TBS, {\rsrew{\PP/\UU(\PP)}}) = \bO(n)$.
Hence, 
$\Rc{\RS}(n) = \comp(n, \TBS, {\rsrew{\PP \cup \UU(\PP)}}) = \bO(n)$
is concluded by Theorem~\ref{t:wgp}.
\end{example}

The next lemma shows that there is no advantage to consider 
SLMIs of dimension $k \geqslant 2$.
\begin{lemma}
If $\RSS$ is compatible with some SLMI
$\A$ then $\RSS$ is compatible with some SLI $\BB$.
\end{lemma}
\begin{proof}
Let $\A$ be an SLMI of dimension $k$.
Further, let $\alpha : \VS \to \N$ denote an arbitrary assignment. 
We define $\widehat{\alpha} \colon \VS \to \N^k$ 
as $\widehat{\alpha}(x) = (\alpha(x),0,\dots,0)^\top$
for each variable $x$.  We define the SLI $\BB$ by 
$f_\BB(x_1,\dots,x_n) = x_1 + \cdots + x_n + \vec{f}_1$. 
Then, 
\begin{align*}
f_\BB(x_1,\dots,x_n) 
  & = \left( 
    (x_1,0,\dots,0)^\top + \cdots + (x_n,0,\dots,0)^\top + \vec{f} 
  \right)_1 \\
  & = \left(
    f_\A((x_1,0,\dots,0)^\top,\dots,(x_n,0,\dots,0)^\top))
    \right)_1
\end{align*}
Therefore, easy structural induction shows that
$\eval{\alpha}{\BB}(t) = (\eval{\widehat\alpha}{\A}(t))_1$
for all terms $t$.  Hence, 
$\RSS \subseteq {\gord{\BB}}$ whenever $\RSS \subseteq {\gord{\A}}$.
\end{proof}

The next example shows that in Lemma~\ref{l:weightgap:i}  SLMIs 
cannot be simply replaced by RMIs.

\begin{example}
\label{ex:7}
Consider the TRSs $\RSexp$
\begin{align*}
 \m{exp}(\mN) & \to \ms(\mN) & \m{d}(\mN) & \to \mN \\
\m{exp}(\m{r}(x)) & \to \m{d}(\m{exp}(x)) & \m{d}(\ms(x)) & \to \ms(\ms(\m{d}(x))) 
\tpkt
\end{align*}
This TRS formalises the exponentiation
function. Setting $t_n = \m{exp}(\m{r}^n(\mN))$ we obtain
$\dheight(t_n, \rsrew{\RSexp}) \geqslant 2^n$ for each $n \geqslant 0$. Thus
the runtime complexity of $\RSexp$ is exponential.

In order to show the claim, we split $\RSexp$ into two TRSs 
$\RS = \{\m{exp}(\mN) \to \ms(0), \m{exp}(\m{r}(x)) \to \m{d}(\m{exp}(x))\}$
and $\RSS = \{\m{d}(\mN) \to \mN, \m{d}(\ms(x)) \to \ms(\ms(\m{d}(x))) \}$.
Then it is easy to verify that the next $1$-dimensional RMI $\A$ 
is compatible with $\RSS$: 
\begin{equation*}
  \mN_{\A} = 0 \qquad \m{d}_{\A}(x) = 3x \qquad \ms_{\A}(x) = x + 1 \tpkt
\end{equation*}
Moreover an upper-bound of $\dheight(t_n ,{\rsrew{\RS/\RSS}})$ 
can be estimated by using the following $1$-dimensional TMI $\BB$:
\begin{equation*}
  \mN_{\BB} = 0 \qquad \m{d}_{\BB}(x) = \ms_{\BB}(x) = x \qquad
  \m{exp}_{\BB}(x) = \m{r}_{\BB}(x) = x + 1 \tpkt
\end{equation*}
%
% $\mN_{\BB} = 0$, $\m{d}_{\BB}(x) = \ms_{\BB}(x) = x$, and 
% $\m{exp}_{\BB}(x) = \m{r}_{\BB}(x) = x + 1$.

Since ${\rsrew{\RS}} \subseteq {\gord[>]{\BB}}$ and 
${\rssrew{\RSS}} \subseteq {\geqord[\geqslant]{\BB}}$ hold, we have 
${\rsrew{{\RS}/{\RSS}}} \subseteq {\gord[>]{\BB}}$.  Hence
$\dheight(t_n, \rsrew{{\RS}/{\RSS}}) \leqslant \eval{\alpha_0}{\BB}(t_n) = n+2$.
But clearly from this we cannot conclude a polynomial bound on the derivation 
length of $\RS \cup \RSS = \RSexp$,
as the runtime complexity of $\RSexp$ is exponential.
\end{example}

Furthermore, non-duplication of $\RS$ is also essential for 
Lemma~\ref{l:weightgap:i}.%
\footnote{This example is due to Dieter Hofbauer and Andreas Schnabl.}

\begin{example}
Consider the following $\RS \cup \SS$
\begin{alignat*}{4}
1\colon &\;&  \mf(\ms(x),y) &\to  \mf(x,\md(y,y,y)) & \qquad
2\colon &\;&  \md(\mN,\mN,x) &\to x 
\\
&& && 3\colon &\;& \md(\ms(x),\ms(y),z) &\to \md(x,y,\ms(z))
\tpkt
\end{alignat*}
Let $\RS = \{1\}$ and let $\SS = \{2,3\}$.
The following SLI $\A$ is compatible with $\SS$:
\begin{equation*}
  \md_\A(x,y,z) = x + y + z + 1 \qquad
  \ms_\A(x) = x + 1 \qquad
  \mN_\A = 0 \tpkt
\end{equation*}
Furthermore, the following 
$\URM{\RS \cup \SS}$-monotone 1-dimensional RMI $\BB$ orients 
the rule in $\RS$ strictly, while the rules in $\SS$ are weakly oriented.
\begin{equation*}
  \mf_\BB(x,y) = x \qquad \md_\BB(x,y,z) = x+y+z \qquad
  \ms_\BB(x) = x + 1 \qquad \mN_\BB = 0 \tpkt
\end{equation*}
Thus, $\comp(n, \TB, {\to_{\RS/\SS}}) = \bO(n)$ is obtained.
If the restriction that $\RS$ is non-duplicating could be dropped
from Lemma~\ref{l:weightgap:i}, we would conclude 
$\Rc{\RS \cup \SS}(n) = \bO(n)$.
However, it is easy to see that $\Rc{\RS \cup \SS}$ is at least 
exponential. Setting $t_n \defsym \mf(\ms^n(\mN),\ms(\mN))$, we
obtain $\dheight(t_n,\rsrew{\RS \cup \RSS}) \geqslant 2^n$ for any $n \geqslant 1$.
\end{example}

We present a weight gap condition for runtime complexity analysis.
When considering the derivation in the beginning of this section 
(on page~\pageref{eq:relative}), every step by a weak dependency pair
only takes place as an outermost step.  Exploiting this fact we
can relax the restriction that was imposed in the above examples.
To this end, we introduce a generalised notion of non-duplicating TRSs.
%NH dropped
%GM perhaps useful in iTRS
% Let $\RS$ be a TRS and ${l \to r} \in {\RS}$.  An \emph{abstraction}
% of $l \rew r$ is a non-duplicating rule $l' \rew r'$ such that 
% $(l' \rew r')\sigma = l \rew r$ for some substitution $\sigma$ mapping 
% variables to variables. 
%
% \begin{definition}
% A TRS $\RS$ is \emph{non-duplicating with respect to an algebra $\A$} 
% (\emph{$\A$-non-duplicating} for short), if for all rules ${l \rew r} \in {\RS}$ 
% there exists an abstraction
% $l' \rew r'$ of $l \rew r$ such that 
% $l \eqord{\A} l'$ and $r \eqord{\A} r'$.  
% \end{definition}

Below 
\( 
\max\,\{\, ([\alpha]_\A(r))_1 \modminus ([\alpha]_\A(l))_1 
\mid \text{$l \to r \in \PP$ and $\alpha : \VV \to \A$} \,\}
\)
is referred to as $\WG(\A,\PP)$.
We say that a $\mu$-monotone RMI is \emph{adequate} if
all compound symbols are interpreted as $\mu$-monotone SLMI.

\begin{lemma}
\label{l:weightgap:ii}
Let $\PP = \WDP(\RS)$ and let $\A$ be an 
adequate $\tmu^{\PP \cup \UU(\PP)}$-monotone RMI.  
Suppose $\WG(\A,\PP)$ is well-defined on $\NN$. % and 
Then, $\PP/\UU(\PP)$ and $\A$ admit a weight gap on $\TBS$.
\end{lemma}
\begin{proof}
The proof follows the proof of Lemma~\ref{l:weightgap:i}.
We set $\Delta = \WG(\A,\PP)$.
Let $s \rsrew{\PP} t$ with $s \in {\to_{\PP \cup \UU(\PP)}}(\TBS)$.
One may write $s = C[l\sigma]$ and $t = C[r\sigma]$ with $l \rew r \in \PP$,
where $C$ denotes a context.
Note that due to $s \in {\to_{\PP \cup \UU(\PP)}}(\TBS)$ all function
symbols above the hole in $C$ are compound symbols.
We perform induction on~$C$.  
\begin{enumerate}
\item
If $C = \Box$ then $[t]_1 - [s]_1 \leqslant \Delta$
by the definition of $\WG(\A,\PP)$.
\item
For inductive step, $C$ must be of the form 
$c(u_1,\ldots,u_{i-1},C',u_{i+1},\ldots,u_n)$ with $i \in \mu(c)$.  
Since $\A$ is adequate, $c_\A$ is a SLMI.  The rest of reasoning is
same with \ref{en:weightgap:i:ii}) in the proof of
Lemma~\ref{l:weightgap:i}.
\end{enumerate}
\end{proof}

\begin{example}[continued from Example~\ref{ex:relative}]
Consider the following adequate $\URM{\PP \cup \UU(\PP)}$-monotone $1$-dimensional
RMI $\BB$:
\begin{xalignat*}{4}
\m{0}_\BB & = \m{c}_\BB = \m{d}_\BB = 0 
& 
\m{s}_\BB(x) & = x + 2
&
\m{-}_\BB(x, y) & = \m{-}^\sharp_\BB(x, y) = {\div}^\sharp_\BB(x, y) = x + 1
\end{xalignat*}
Since $\Delta(\BB,\PP)$ is well-defined
(indeed $1$), $\BB$ admits the weight gap of Lemma~\ref{l:weightgap:ii}.
Moreover, $\UU(\PP)$ is compatible with ${\gord{\BB}}$.
As $\comp(n, \TBS, {\rsrew{\PP/\UU(\PP)}}) = \bO(n)$ 
was shown in Example~\ref{ex:relative},
Theorem~\ref{t:wgp} deduces linear runtime complexity for $\RSdiv$.
\end{example}

In Lemma~\ref{l:weightgap:ii} $\WG(\A,\PP)$ must be well-defined.

\begin{example}
Consider the following TRS $\RS$
\begin{alignat*}{4}
1\colon~ && \m{f}([\,]) & \to [\,] & 
3\colon~ && \m{g}([\,], z) & \to z \\
2\colon~ && \m{f}(x : y) & \to x : \m{f}(\m{g}(y, [\,])) \qquad &
4\colon~ && \m{g}(x : y, z) & \to \m{g}(y, x : z) 
\end{alignat*}
whose optimal innermost runtime complexity is quadratic.
The weak innermost dependency pairs $\PP \defsym \WIDP(\RS)$ are
\begin{alignat*}{4}
5\colon~ && \m{f}^\sharp([\,]) & \to \m{c} & 
7\colon~ && \m{g}^\sharp([\,], z) & \to \m{d} \\
6\colon~ && \m{f}^\sharp(x : y) & \to \m{f}^\sharp(\m{g}(y, [\,])) \qquad &
8\colon~ && \m{g}^\sharp(x : y, z) & \to \m{g}^\sharp(y, x : z) 
\end{alignat*}
and $\UU(\PP) = \{3,4\}$.
It is not difficult to show
$\comp(n, \TBS, {\rsirew{\PP/\UU(\PP)}}) = \bO(n)$
with a $1$-dimensional RMI.  Moreover, the 
$\IURM{\PP \cup \UU(\PP)}$-monotone $1$-dimensional RMI $\A$ with
\begin{align*}
[\,]_\A & = 0 &
{:}_\A(x,y) & = y + 1 &
\m{g}_\A(x,y) & = 2x + y + 1  \\
\m{f}_\A(x) & = \m{f}^\sharp_\A(x) = x &
\m{g}^\sharp_\A(x,y) & = 0 &
\m{c}_\A & = \m{d}_\A = 0
\end{align*}
is compatible with $\UU(\PP)$.
If Lemma~\ref{l:weightgap:ii} would be applicable without its
well-definedness, linear innermost runtime complexity of $\RR$ 
would be concluded falsely.
Note that $\WG(\A,\PP)$ is \emph{not} well-defined on $\NN$ 
due to pair 6.
\end{example}

\begin{corollary} \label{c:main}
Let $\RS$ be a TRS, $\PP$ the set of weak (innermost) dependency
pairs, and $\mu$ be the (innermost) usable replacement map.
Suppose $\BB$ is a RMI such that $(\geqord{\BB},\gord{\BB})$ forms a 
$\mu$-complexity pair with
$\UU(\PP) \subseteq {\geqord{\BB}}$ and $\PP \subseteq {\gord{\BB}}$.
Further, suppose $\A$ is an adequate $\mu$-monotone RMI
such that $\WG(\A,\PP)$ is well-defined on $\NN$ and $\PP$ is 
compatible with $\UU(\PP)$. 

Then the (innermost) runtime complexity function $\rc^{(\m{i})}_{\RS}$
with respect to $\RS$ is polynomial. 
Here the degree of the polynomial is given by 
the maximum of the degrees of the used RMIs.
\end{corollary}

Let $\A$ be an RMI as in the corollary. In order to verify
that $\WG(\A,\PP)$ is well-defined, we use
the following simple trick in the implementation. 
Let  $l \to r \in \PP$ and let $k$ denotes the cardinality of $\Var(l) \supseteq \Var(r)$.
Recall the existence of matrices (over $\N$) 
$L_1,\dots,L_k$, $R_1,\dots,R_k$ and vectors $\vec{l}$, $\vec{r}$ such that
$   
\eval{\alpha}{\A}(l) \modminus \eval{\alpha}{\A}(r) =
   \sum_{i=1}^k (R_i \modminus L_i) \alpha(x_i) + (\vec{r} \modminus \vec{l})
$. 
Then $\WG(\A,\PP)$ is well-defined if $(R_i \modminus L_i) \leqslant \mathbf{0}$.

\section{Weak Dependency Graphs} \label{DG}

In this section we extend the above refinements
by revisiting dependency graphs in the context
of complexity analysis. Let $\PP = \WDP(\RSdiv)$ and recall the 
derivation over $\PP \cup \UU(\PP)$ on page~\pageref{ex:relative}.
Looking more closely at this derivation we observe that we
do not make use of all weak dependency pairs in $\PP$, but we
only employ the pairs $7$ and $8$:
\label{eq:wdg}
\begin{equation*}
\mathbf{4} \div^\sharp \mathbf{2} 
~\rsrew{\{8\}/\UU(\PP)}~ \mathbf{2} \div^\sharp \mathbf{2}
~\rsrew{\{8\}/\UU(\PP)}~ \m{0} \div^\sharp \mathbf{2}
~\rsrew{\{7\}/\UU(\PP)}~ \m{c} \tpkt
\end{equation*}
Therefore it is a natural idea to modularise our complexity
analysis and apply the previously obtained techniques only
to those pairs that are relevant.
Dependencies among weak dependency pairs are formulated by the notion of weak dependency graphs,
which is an easy variant of \emph{dependency graphs}~\cite{ArtsGiesl:2000}.

\begin{definition} \label{d:DG}
Let $\RS$ be a TRS over a signature $\FS$ and 
let $\PP$ be the set of weak, weak
innermost, or (standard) dependency pairs. 
The nodes of the \emph{weak dependency graph} $\WDG(\RS)$, 
\emph{weak innermost dependency graph} $\WIDG(\RS)$, or
\emph{dependency graph} $\DG(\RS)$
are the elements of $\PP$ and there is an arrow from $s \to t$ to $u \to v$ if and
only if there exist a context $C$ and substitutions 
$\sigma, \tau \colon \VV \to \TT(\FS, \VV)$ such that
$t\sigma \rew^* C[u\tau]$, where $\rew$ denotes 
$\rsrew{\RS}$ or $\irew{\RS}$
depending on whether $\PP = \WDP(\RS)$, $\PP = \DP(\RS)$, or 
$\PP = \WIDP(\RS)$, respectively.
\end{definition}

\begin{example}[continued from Example~\ref{ex:1:WDP}]
\label{ex:1:wdg}
The weak dependency graph $\WDG(\RSdiv)$ has the following form.  
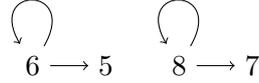
\begin{center}
\begin{tikzpicture}[node distance=5mm]
\node(6) {6} ;
\node(5) [right=of 6] {5} ;

\draw[->] (6) edge (5) ;
\draw[->] (6) edge [in=120,out=60,loop] (6) ;

\node(8) [right=of 5] {8} ;
\node(7) [right=of 8] {7} ;

\draw[->] (8) edge (7) ;
\draw[->] (8) edge [in=120,out=60,loop] (8) ;
\end{tikzpicture}
\end{center}  
\end{example}

Since weak dependency graphs represent call graphs of functions, 
grouping mutual parts helps analysis.
A graph is called \emph{strongly connected} if any node is connected with every other
node by a (possibly empty) path. 
A \emph{strongly connected component} (\emph{SCC} for short) is a
maximal strongly connected subgraph.%
\footnote{%
We use SCCs in the standard graph theoretic sense, while
in the literature SCCs are sometimes defined as \emph{maximal cycles} 
(e.g. \cite{GAO:2002,HirokawaMiddeldorp:2005,T07}). 
This alternative definition is of limited use in our context.}

\begin{definition} \label{d:1}
Let $\GG$ be a graph, let $\equiv$ denote the equivalence relation induced by SCCs, 
and let $\PP$ be a SCC in $\GG$. Consider the
\emph{congruence graph} $\PG{\GG}$ induced by the equivalence relation $\equiv$.
The set of all source nodes in $\PG{\GG}$ is denoted by $\Src(\PG{\GG})$.
Let $\KK \in \PG{\GG}$ and
let $\CC$ denote the SCC represented by $\KK$. 
Then we write $l \to r \in \KK$ if $l \to r \in \CC$.
For nodes $\KK$ and $\LL$ in $\PG{\GG}$ we write $\KK \edge \LL$, if
$\KK$ and $\LL$ are connected by an edge. The reflexive (transitive, reflexive-transitive)
closure of $\edge$ is denoted as $\redge$ ($\tedge$, $\rtedge$).
\end{definition}

\begin{example}[continued from Example~\ref{ex:1:wdg}]
Let $\GG$ denote $\WDG(\RSdiv)$.  
There are 4 SCCs in $\GG$: %,
$\{5\}$, $\{6\}$, $\{7\}$, and $\{8\}$.
Thus the congruence graph $\PG{\GG}$ has the following form:  
\begin{center}
\begin{tikzpicture}[node distance=5mm]
\node(6) {6} ;
\node(5) [right=of 6] {5} ;
\draw[->] (6) edge (5) ;

\node(8) [right=of 5] {8} ;
\node(7) [right=of 8] {7} ;
\draw[->] (8) edge (7) ;
\end{tikzpicture}
\end{center}  
Here $\Src(\PG{\GG}) = \{ \{ 6 \}, \{ 8 \} \}$.
\end{example}

\begin{example} \label{ex:2}
Consider the TRS $\RSgcd$ which computes the greatest common divisor.%
\footnote{This is Example~3.6a in Arts and Giesl's collection of TRSs~\cite{ArtsGiesl:2001}.}
\begin{alignat*}{4}
1\colon && \m{0} \leqslant y & \to \m{true} & 
6\colon && \m{gcd}(\m{0}, y) & \to y \\
2\colon && \m{s}(x) \leqslant \m{0} & \to \m{false} & 
7\colon && \m{gcd}(\m{s}(x), \m{0}) & \to \m{s}(x) \\
3\colon && \m{s}(x) \leqslant \m{s}(y) & \to x \leqslant y 
& \hspace{3ex}
8\colon && \m{gcd}(\m{s}(x), \m{s}(y)) 
& \to \m{if_{gcd}}(y \leqslant x, \m{s}(x), \m{s}(y))\\
4\colon && x - \m{0} & \to x &
9\colon && \m{if_{gcd}}(\m{true}, \m{s}(x), \m{s}(y)) 
& \to \m{gcd}(x - y, \m{s}(y)) \\
5\colon && \m{s}(x) - \m{s}(y) & \to x - y &
10\colon && \m{if_{gcd}}(\m{false}, \m{s}(x), \m{s}(y)) 
& \to \m{gcd}(y - x, \m{s}(x))
\tpkt
\intertext{The set $\WDP(\RSgcd)$ consists of the next ten weak dependency pairs:}
11\colon && \m{0} \leqslant^\sharp y & \to \m{c_1} 
& \hspace{3ex}
16\colon && \m{gcd}^\sharp(\m{0}, y) & \to y
\\
12\colon && \m{s}(x) \leqslant^\sharp \m{0} & \to \m{c_2} &
17\colon && \m{gcd}^\sharp(\m{s}(x), \m{0}) & \to x  
\\
13\colon && \m{s}(x) \leqslant^\sharp \m{s}(y) & \to x \leqslant^\sharp y &
18\colon && \m{gcd}^\sharp(\m{s}(x), \m{s}(y)) & \to \m{if_{gcd}}^\sharp(y \leqslant x, \m{s}(x), \m{s}(y))
\\
14\colon && \m{s}(x) -^\sharp \m{0} & \to x &
19\colon && \m{if_{gcd}}^\sharp(\m{true}, \m{s}(x), \m{s}(y)) 
& \to \m{gcd}^\sharp(x - y, \m{s}(y))
\\
15\colon && \m{s}(x) -^\sharp \m{s}(y) & \to x -^\sharp y &
20\colon && \m{if_{gcd}}^\sharp(\m{false}, \m{s}(x), \m{s}(y)) & \to \m{gcd}^\sharp(y - x, \m{s}(x))
\tpkt
\end{alignat*}
The congruence graph $\PG{\GG}$ of $\GG \defsym \WDG(\RSgcd)$ 
has the following form:
\begin{center}
\begin{tikzpicture}[node distance=5mm]
\node(11) {11} ;
\node(13) [right=of 11] {13} ;
\node(12) [right=of 13] {12} ;
\draw[->] (13) edge (12) ;
\draw[->] (13) edge (11) ;
\node(15) [right=of 12] {15} ;
\node(14) [right=of 15] {14} ;
\draw[->] (15) edge (14) ;
\node(18) [right=of 14] {\{18,19,20\}} ;
\node(16) [right=of 18] {16} ;
\node(17) [right=of 16] {17} ;
\draw[->] (18) edge (16) ;
\end{tikzpicture}
\end{center}
Here $\Src(\PG{\GG}) = \{ \{ 13 \}, \{ 15 \}, \{17\}, \{18,19,20\} \}$.
\end{example}

The main result in this section is stated as follows:
Let $\RS$ be a TRS, $\PP = \WDP(\RS)$, $\GG = \WDG(\RS)$, and furthermore
\label{eq:path}
\begin{equation*}
\Path(t) \defsym \max\{ \dheight(t, \rsparenirew{\QQ \cup \UU(\QQ)}) \mid 
\text{$(\PP_1,\ldots,\PP_k)$ is a path in $\PG{\GG}$ and $\PP_1 \in \Src(\PG{\GG})$} \} \tkom
\end{equation*}
where $\QQ = \bigcup_{i=1}^k \PP_i$.    Then, 
\(
\dheight(t,{\rsrew{\RS}}) = \bO(\Path(t)) 
\)
holds for all basic term $t$.  This means that one may decompose 
$\PP \cup \UU(\PP)$ into several smaller fragments and analyse these
fragments separately.

Reconsider the derivation on page~\pageref{eq:wdg}.
The only dependency pairs are from the set 
$\{7,8\}$. Observe that the order these pairs are applied is 
representable by the path 
$(\{8\},\{7\})$ in the congruence graph. 
This observation is cast into the following definition.

\begin{definition}
\label{d:pathbased}
Let $\PP$ be the set of weak (innermost) dependency
pairs and let $\GG$ denote the weak (innermost) dependency graph.
Suppose $A \colon {s} \rsparenisrew{\PP/\UU(\PP)} {t}$ denote a derivation, 
such that $s \in \TBS$. 
If $A$ can be written in the following form:
\begin{equation*}
  {s} \rsparenisrew{\PP_1/\UU(\PP)} \cdots \rsparenisrew{\PP_k/\UU(\PP)} {t} \tkom
\end{equation*}
then $A$ is \emph{based on the sequence of nodes $(\PP_1,\ldots,\PP_k)$ (in $\PG{\GG}$)}.
\end{definition}

The next lemma is an easy generalisation of the above example.

\begin{lemma}
\label{l:15}
Let $\RS$ be a TRS, let $\PP$ be the set of weak (innermost) dependency
pairs and let $\GG$ denote the weak (innermost) dependency graph. 
Suppose that all compound symbols are nullary. Then
any derivation $A \colon {s} \rsparenisrew{\PP/\UU(\PP)} {t}$ such that $s \in \TBS$
is based on a path in~$\PG{\GG}$.
\end{lemma}

From Lemma~\ref{l:15} we see that the above mentioned modularity result
easily follows as long as the arity of the compound symbols is restricted.
We lift the assumption that all compound symbols are nullary.
Perhaps surprisingly this generalisation complicates the matter. As
exemplified by the next example, Lemma~\ref{l:15} fails if there 
exist non-nullary compound symbols.

\begin{example} \label{ex:dg:1}
Consider the TRS 
$\RS = \{\m{f}(\m{0}) \to \m{a}, 
 \m{f}(\m{s}(x)) \to \m{b}(\m{f}(x), \m{f}(x))\}$. %,
The set $\WDP(\RS)$ consists of the two weak dependency pairs:
$1\colon \m{f}^\sharp(\m{0}) \to \m{c}$ and
$2\colon \m{f}^\sharp(\m{s}(x)) \to \m{d}(\m{f}^\sharp(x), \m{f}^\sharp(x))$.
The corresponding congruence graph only contains the single edge from
$\{2\}$ to $\{1\}$.
Writing $t_n$ for $\m{f}^\sharp(\m{s}^n(\m{0}))$, we have the sequence
\begin{align*}
t_2 & \to_{\{2\}}^2 \m{d}(\m{d}(t_0, t_0), t_1) 
\rsrew{\{1\}}   \m{d}(\m{d}(\m{c}, t_0), t_1) \\
    & \rsrew{\{2\}}   \m{d}(\m{c}(\m{c}, t_0), \m{d}(t_0, t_0)) 
\to_{\{1\}}^3 \m{d}(\m{d}(\m{c}, \m{c}), \m{d}(\m{c},\m{c})) \tpkt
\end{align*}
whereas $(\{2\},\{1\},\{2\},\{1\})$ is not a path in the graph.
\end{example}

Note that the derivation in Example~\ref{ex:dg:1} can be
reordered (without affecting its length) such that the derivation becomes
based on the path $(\{2\},\{1\})$. More generally, we observe 
that a weak (innermost) dependency pair containing an $m$-ary ($m > 1$) 
compound symbol can induce $m$ \emph{independent} derivations. 
This allows us to reorder (sub-)derivations. We show this via the following sequence of lemmas.

Let $\RS$ be a TRS, let $\PP$ denote the set of weak (innermost) dependency pairs,
and let $\GG$ denote the weak (innermost) dependency graph.
The set $\TBC$ is inductively defined as follows (i) $\TTs \cup \TT \subseteq \TBC$,
where $\TTs = \{t^{\sharp} \mid t \in \TT \}$ and
(ii) $c(t_1,\ldots,t_n) \in \TBC$, whenever $t_1,\ldots,t_n \in \TBC$ and $c$ a compound symbol.
The next lemma formalises an easy observation.
\begin{lemma} \label{l:12}
Let $\CC$ be a set of nodes in $\GG$
and let $A \colon {t = t_0} \rsparenisrew{\CC/\UU(\PP)} {t_n}$ denote a
derivation based on $\CC$ with $t \in \TBC$. 
Then $A$ has the following form:
$t = t_{0} \rsparenirew{\CC/\UU(\PP)} t_{1} 
\rsparenirew{\CC/\UU(\PP)} \dots 
\rsparenirew{\CC/\UU(\PP)} t_{n}$
where each $t_i \in \TBC$. 
\end{lemma}

A key is that consecutive two weak dependency pairs may be swappable.
\begin{lemma} \label{l:13}
Let $\KK$ and $\LL$ denote two different nodes in $\PG{\GG}$ such that
there is no edge from $\KK$ to $\LL$. Let $s \in \TBC$ and suppose the existence of a 
derivation $A$ of the following form:
\begin{equation*}
  {s} \rsparenirew{\KK/\UU(\PP)} \cdot \rsparenirew{\LL/\UU(\PP)} t \tpkt
\end{equation*}
Then there exists a derivation $B$ 
\begin{equation*}
  {s} \rsparenirew{\LL/\UU(\PP)} \cdot \rsparenirew{\KK/\UU(\PP)} {t} \tkom
\end{equation*}
such that $\card{A} = \card{B}$.
\end{lemma}
\begin{proof}
We only show the full rewriting case since the innermost case is analogous.
According to Lemma~\ref{l:12} an arbitrary terms $u$ reachable from $s$ 
belongs to $\TBC$.  Writing 
$\SC{u_1,\ldots,u_i,\ldots,u_m}_{\FS \cup \FS^\sharp}$ for $u$, 
the $m$-hole context $C$ consists of compound symbols and 
variables, $u_1,\ldots,u_m \in \TT \cup \TT^\sharp$.
Therefore, $A$ can be written in the following form:
\begin{alignat*}{3}
s 
&  ~\to_{\UU(\PP)}^{n_1}~
&& \SC{u_1,\ldots,u_i,\ldots,u_m}_{\FS \cup \FS^\sharp} 
&& =: u
\\
&  ~\to_\LL~
&& C[u_1,\ldots,u_i',\ldots,u_m] 
\\
&  ~\to_{\UU(\PP)}^{n_2}~
&& C[v_1,\ldots,v_i,\ldots,v_j,\ldots,v_m]
\\
&  ~\to_\KK~
&& C[v_1,\ldots,v_i,\ldots,v_j',\ldots,v_m] 
&& ~\to_{\UU(\PP)}^{n_3}~ t \tkom
\end{alignat*}
with $u_i' \to_{\UU(\PP)}^k v_i$. 
Here $i \neq j$ holds, because $i = j$ induces $\LL \leadsto \KK$.
Easy induction on $n_2$ shows
\begin{alignat*}{2}
s  
& ~\to_{\UU(\PP)}^{n_1}~ u ~=~ 
&& C[u_1,\ldots,u_i,\ldots,u_j,\ldots,u_m]
\\
&  ~\to_{\UU(\PP)}^{n_2 - k}~
&& C[v_1,\ldots,u_i,\ldots,v_j,\ldots,v_m]
\\
&  ~\to_\KK~
&& C[v_1,\ldots,u_i,\ldots,v_j',\ldots,v_m]
\\
&  ~\to_\LL~
&& C[v_1,\ldots,u_i',\ldots,v_j',\ldots,v_m]
\\
&  ~\to_{\UU(\PP)}^k~
&& C[v_1,\ldots,v_i,\ldots,v_j',\ldots,v_m] 
   ~\to_{\UU(\PP)}^{n_3} t~ \tkom
\end{alignat*}
which is the desired derivation $B$.  
\end{proof}

The next lemma states that reordering is partly possible.
\begin{lemma} \label{l:14}
Let $s \in \TBC$, and let $A \colon {s} \rsparenisrew{\PP/\UU(\PP)} {t}$ be a 
derivation based on a sequence of nodes $(\PP_1,\ldots,\PP_k)$ such that 
$\PP_1 \in \Src(\PG{\GG})$, and let $(\QQ_1,\ldots,\QQ_{\ell})$ be 
a path in $\PG{\GG}$ with $\{\PP_1,\dots,\PP_k\} = \{\QQ_1,\dots,\QQ_\ell\}$.
Then there exists a derivation 
$B \colon {s} \rsparenisrew{\PP/\UU(\PP)} {t}$ based on $(\QQ_1,\ldots,\QQ_{\ell})$
such that $\card{A} = \card{B}$ and $\PP_1 = \QQ_1$.
\end{lemma}
\begin{proof}
According to Lemma~\ref{l:12}, for any derivation $A$
\begin{equation*}
s \rsparenisrew{\PP_1/\UU(\PP)} \cdots \rsparenisrew{\PP_n/\UU(\PP)} t \tkom
\end{equation*}
if $\PP_i \edge \PP_{i+1}$
does not hold, there is a derivation $B$
\begin{equation*}
s \rsparenisrew{\PP_1/\UU(\PP)} \cdots 
  \rsparenisrew{\PP_{i+1}/\UU(\PP)} \cdot
  \rsparenisrew{\PP_i/\UU(\PP)} \cdots 
  \rsparenisrew{\PP_n/\UU(\PP)} t
\tkom
\end{equation*}
with $\card{A} = \card{B}$. By assumption $(\QQ_1,\ldots,\QQ_{\ell})$
is a path, whence we obtain $\QQ_1 \edge \cdots \edge \QQ_{\ell}$.
By performing bubble sort with respect to $\tedge$, $A$ is transformed into the derivation $B$:
\begin{equation*}
s \rsparenisrew{\QQ_1/\UU(\PP)} \cdots 
  \rsparenisrew{\QQ_m/\UU(\PP)} t \tkom
\end{equation*}
such that $\card{A} = \card{B}$.
\end{proof}

The next example shows that there is a derivation that
cannot be transformed into a derivation based on a path.

\begin{example} \label{ex:dg:2}
Consider the TRS 
$\RS = \{\m{f} \to \m{b}(\m{g},\m{h}), \m{g} \to \m{a}, \m{h} \to \m{a}\}$.
Thus $\WDP(\RS)$ consists of three dependency pairs:
$1\colon \m{f}^\sharp \to \m{c}(\m{g}^\sharp,\m{h}^\sharp)$, 
$2\colon \m{g}^\sharp \to \m{d}$, and
$3\colon \m{h}^\sharp \to \m{e}$.
Let $\PP \defsym \WDP(\RS)$ and let $\GG \defsym \WDG(\RS)$. 
Note that $\PG{\GG}$ are identical to $\GG$.
We witness that the derivation
\begin{equation*}
\m{f}^\sharp 
\rsrew{\PP} \m{c}(\m{g}^\sharp, \m{h}^\sharp)
\rsrew{\PP} \m{c}(\m{d}, \m{h}^\sharp)
\rsrew{\PP} \m{c}(\m{d}, \m{e}) \tkom
\end{equation*}
is based neither on the path $(\{1\},\{2\})$,
nor on the path $(\{1\},\{3\})$.
\end{example}

Lemma~\ref{l:14} shows that we can reorder a given derivation 
$A$ that is based on a sequence of nodes that would in principle
form a path in the congruence graph $\PG{\GG}$. The next
lemma shows that we can guarantee that any derivation 
is based on sequence of different paths.

\begin{lemma}
\label{l:21}
Let $s \in \TBC$ and let $A \colon {s} \rsparenisrew{\PP/\UU(\PP)} {t}$ 
be a derivation based on
$(\PP_1,\ldots,\PP_k, \QQ_1,\ldots,\QQ_\ell)$, such that
$(\PP_1,\ldots,\PP_k)$ and $(\QQ_1,\ldots,\QQ_\ell)$ form 
two disjoint paths in $\GG$. Then there exists a derivation 
$B \colon {s} \rsparenisrew{\PP/\UU(\PP)} {t}$ 
based on the sequence of nodes 
$(\QQ_1,\ldots,\QQ_\ell,\PP_1,\ldots,\PP_k)$ such 
that $\card{A} = \card{B}$.
\end{lemma}
\begin{proof}
The lemma follows by an adaptation of the technique in the proof of 
Lemma~\ref{l:14}.
\end{proof}

Lemma~\ref{l:21} shows that the maximal length of any derivation 
only differs from the maximal length of any derivation based on a path 
by a linear factor, 
depending on the size of the congruence graph $\PG{\GG}$.
We arrive at the main result of this section. Recall the
definition of $\Path(\cdot)$ on page~\pageref{eq:path}.
\begin{theorem} \label{t:dg}
Let $\RS$ be a TRS and $\PP$ the set of weak (innermost) dependency pairs.
Then, $\dheight(t, \rsparenirew{\RS}) = \bO(\Path(t))$ holds for all $t \in \TBS$.
\end{theorem}
\begin{proof}
Let $a$ denotes the maximum arity of compound symbols and 
$K$ denotes the number of SCCs in the weak (innermost) dependency graph $\GG$.
We show $\dheight(s, \rsparenirew{\RS}) \leqslant a^{K} \cdot \Path(s)$
holds for all $s \in \TBS$.
Theorem~\ref{t:dp:usable} yields that ${\dheight(s,\rsparenirew{\RS})} = {\dheight(s,\rew)}$,
where $\rew$ either denotes $\rsrew{\PP \cup \UU(\PP)}$ or $\irew{\PP \cup \UU(\PP)}$. 

Let $A\colon {s} \rss {t}$ be a derivation over $\PP \cup \UU(\PP)$ 
such that $s \in \TBS$. 
Then $A$ is based on a sequence of nodes in the congruence graph $\PG{\GG}$
such that there exists a maximal (with respect to subset inclusion) components of
$\PG{\GG}$ that includes all these nodes. Let $T$ denote this maximal component.
$T$ forms a directed acyclic graph. In order to (over-)estimate the number of
nodes in this graph we can assume without loss of
generality that $T$ is a tree with root in $\Src(\PG{\GG})$. 
Note that $K$ bounds the height of this tree. 
Thus the number of nodes in the component $T$ is less than 
\begin{equation*}
  \frac{a^{K} - 1}{a-1} \leqslant a^K \tpkt
\end{equation*}
Due to Lemma~\ref{l:21} the derivation $A$ is conceivable as a sequence of
subderivations based on paths in  $\PG{\GG}$. As the number of nodes in $T$
is bounded from above by $a^K$, there 
exist at most be $a^K$ different paths through $T$. 

Hence in order to estimate $\card{A}$, 
it suffices to estimate the length of any subderivation $B$ of $A$, based on a specific path.
Let $(\PP_1,\ldots,\PP_k)$ be a path in $\PG{\PP}$ such that $\PP_1 \in \Src(\PG{\GG})$
and let $B \colon u \rew^n v$, denote a derivation based on this path. 
Let $\QQ \defsym \bigcup_{i=1}^k \PP_i$. 
By Definition~\ref{d:pathbased} and the definition of usable rules, the derivation $B$ 
can be written as:
\begin{equation*} 
u=u_0 \rsparenirew{\PP_1/\UU(\QQ)}
u_{n_1} \rsparenirew{\PP_2/{\UU(\QQ)}}
\cdots \rsparenirew{\PP_k/{\UU(\QQ)}} u_n = v \tkom
\end{equation*}
where $u \in \TBS$ each $u_i \in \TBC$. Hence $B$ is contained in $u \rsparenisrew{\QQ \cup \UU(\QQ)} v$
and thus $\card{B} \leqslant \Path(u)$ by definition.

As the length of a derivation $B$ based on a specific path can 
be estimated by $\Path(s)$, 
we obtain that the length of an arbitrary derivation is less than
$a^K \cdot \Path(s)$.
This completes the proof of the theorem.
\end{proof}

\begin{corollary}
\label{c:dg}
Let $\RS$ be a TRS and let $\GG$ denote the weak (innermost) dependency graph.
For every path $\bar{P} \defsym (\PP_1,\ldots,\PP_k)$ in 
$\PG{\GG}$ such that $\PP_1 \in \Src(\PG{\GG})$,
we set $\QQ \defsym \bigcup_{i=1}^k \PP_i$ and suppose
\begin{enumerate}
\item there exist a 
$\URM{\QQ \cup \UU(\QQ)}$-monotone 
($\IURM{\QQ \cup \UU(\QQ)}$-monotone)
and adequate RMI $\A_{\bar{P}}$
that admits the weight gap 
$\EWG(\A_{\bar{P}},\QQ)$ on $\TBS$ and 
$\A_{\bar{P}}$ is compatible with the usable rules $\UU(\QQ)$,

\item there exists a $\URM{\QQ \cup \UU(\QQ)}$-monotone
($\IURM{\QQ \cup \UU(\QQ)}$-monotone) RMI $\BB_{\bar{P}}$ such that 
$(\geqord{\BB_{\bar{P}}},\gord{\BB_{\bar{P}}})$ forms a complexity pair for 
$\PP_k/{\PP_1 \cup \cdots \cup \PP_{k-1} \cup \UU(\QQ)}$, and
%NH dropped
%\item $\PP_k$ is $\A_{\bar{P}}$-non-duplicating.
\end{enumerate}
Then the (innermost) runtime complexity of a TRS $\RS$ is polynomial.
Here the degree of the polynomial is given by 
the maximum of the degrees of the used RMIs.
\end{corollary}
\begin{proof}
We restrict our attention to weak dependency pairs
and full rewriting. First observe that the assumptions imply that any basic term $t \in \TB$
is terminating with respect to $\RS$. Let $\PP$ be the set
of weak dependency pairs. (Note that $\PP \supseteq \QQ$.)
By Lemma~\ref{l:5} any infinite
derivation with respect to $\RS$ starting in $t$ can be translated into
an infinite derivation with respect to $\UU(\PP) \cup \PP$.
Moreover, as the number of paths in $\PG{\GG}$ is finite, there 
exist a path $(\PP_1,\ldots,\PP_k)$ in $\PG{\GG}$ and an infinite rewrite sequence
based on this path. This is a contradiction. Hence we can employ Theorem~\ref{t:wgp}
in the following.

Let $(\PP_1,\ldots,\PP_k)$ be an arbitrary, but fixed path in 
the congruence graph $\PG{\GG}$, let $\QQ = \bigcup_{i=1}^k \PP_i$, 
and let $d$ denote the  maximum of the degrees of the used RMIs.  
Due to Theorem~\ref{t:wgp} there exists $c \in \N$ such that:
\begin{equation*}
  \dheight(t^\sharp,\rsrew{\QQ \cup \UU(\QQ)}) \leqslant
  (1 + \EWG(\A_{\bar{P}},\QQ)) \cdot \dheight(t^\sharp,\rsrew{\QQ/\UU(\QQ)}) + 
c \cdot \size{t}^d \tpkt
\end{equation*}
Due to Theorem~\ref{t:dg} it suffices to consider a
derivation $A$ based on the path $(\PP_1,\ldots,\PP_k)$.
Suppose $A \colon s \rsnrew{\QQ/\UU(\QQ)}{n} t$. Then $A$ can be 
represented as follows:
\begin{equation*} 
s=s_0 \rsnrew{\PP_1/\UU(\PP_1)}{n_1}
s_{n_1} \rsnrew{\PP_2/{\UU(\PP_1) \cup \UU(\PP_2)}}{n_2}
\cdots \rsnrew{\PP_k/{\UU(\PP_1) \cup \cdots \cup \UU(\PP_k)}}{n_k} s_n = t \tkom
\end{equation*}
such that $n = \sum_{i=1}^k n_i$. 
It is sufficient to bound each $n_i$ from the above.
Fix $i \in \{1,\dots,k\}$. Consider the subderivation
\begin{equation*}
A'\colon
s=s_0 \rsnrew{\PP_1/\UU(\PP_1)}{n_1} s_{n_1} \cdots 
\rsnrew{\PP_k/{\UU(\PP_1) \cup \cdots \cup \UU(\PP_i)}}{n_i} s_{n_i}
\tpkt
\end{equation*}
Then $A'$ is contained in 
$A'' \colon 
s \rssrew{\PP_1 \cup \cdots \cup \PP_{i-1} \cup 
    \UU(\PP_1) \cup \cdots \UU(\PP_i)} \cdot \rsnrew{\PP_k/{\UU(\PP_1) \cup \cdots \cup \UU(\PP_i)}}{n_i} s_{n_i}$. 
Let $\hat{P_i} \defsym (\PP_1, \ldots, \PP_{i})$. By assumption there 
exists a $\mu$-monotone complexity pair $(\geqord{\BB_{\hat{P_i}}},\gord{\BB_{\hat{P_i}}})$ 
such that
$\PP_1 \cup \cdots \cup \PP_{i-1} \cup \UU(\PP_1 \cup \cdots \cup \PP_i)
\subseteq {\geqord{\BB_{\hat{P_i}}}}$ and $\PP_i \subseteq {\gord{\BB_{\hat{P_i}}}}$.
Hence, we obtain $n_i \leqslant (\eval{\alpha_0}{\BB_{\hat{P_i}}}(s))_1$ and in sum 
${n} \leqslant {k \cdot \size{s}^d}$. 
Finally, defining the polynomial $p$ as follows:
\begin{equation*}
  p(x) \defsym (1 + \EWG(\A_{\bar{P}},\QQ)) \cdot k \cdot x^d + c \cdot x^d \tkom
\end{equation*}
we conclude 
$\dheight(t^\sharp, \rsrew{\QQ \cup \UU(\QQ)}) \leqslant p(\size{t})$. Note
that the polynomial $p$ depends only on the algebras $\A_{\bar{P}}$ and 
$\BB_{\hat{P_1}}$, \dots, $\BB_{\bar{P_k}}$.

As the path $(\PP_1,\ldots,\PP_k)$ was chosen arbitrarily, there exists
a polynomial $q$, depending only on the employed RMIs such that
$\Path(t) \leqslant q(\size{t})$. Thus the
corollary follows due to Theorem~\ref{t:dg}.
\end{proof}

Let $t$ be an arbitrary term. By definition the set in $\Path(t)$ may
consider $2^{\bO(n)}$-many paths, where $n$ denotes the number of nodes in $\PG{\GG}$.
However, it suffices to restrict the definition on page~\pageref{eq:path}
to \emph{maximal} paths. For this refinement $\Path(t)$ contains at most $n^2$ paths.
This fact we employ in implementing the WDG method.

\begin{example}[continued from Example~\ref{ex:2}] 
For $\PG{\WDG(\RSgcd)}$ the above set consists of 8 paths:
$(\{13\})$, $(\{13\},\{11\})$, $(\{13\},\{12\})$, 
$(\{15\})$, $(\{15\},\{14\})$, 
$(\{17\})$,
$(\{18,19,20\})$, and $(\{18,19,20\},\{16\})$. 
 In the following we only consider the last three paths,
since all other paths are similarly handled.
\begin{itemize}
\item
Consider $(\{17\})$.  Note $\UU(\{17\}) = \varnothing$.
By taking an arbitrary SLI $\A$ and the linear restricted interpretation
$\BB$ with 
$\m{gcd}^\sharp_\BB(x,y) = x$ and $\m{s}_\BB(x) = x + 1$, we have
$\varnothing \subseteq {>_\A}$,
$\varnothing \subseteq {\geqslant_\BB}$, and 
$\{17\} \subseteq {>_\BB}$.
\item
Consider $(\{18,19,20\})$. Note $\UU(\{18,19,20\}) = \{1,\ldots,5\}$.
The following RMI $\A$ is adequate for $(\{18,19,20\})$ and
strictly monotone on $\URM{\PP \cup \UU(\PP)}$. The presentation
of $\A$ is succinct as only the signature of the 
usable rules $\{1,\ldots,5\}$ is of interest.
\begin{align*}
\m{true}_\A &= \m{false}_\A = \m{0}_\A = \vec{0} 
&
\ms_\A(\vec{x}) & =
  \begin{pmatrix}
    1 & 1 \\
    0 & 1
  \end{pmatrix}
  \vec x
  +
  \begin{pmatrix}
    3\\
    1
  \end{pmatrix}
\\
{\leqslant}_\A(\vec{x}, \vec{y})
&= 
  \begin{pmatrix}
    0 & 1\\
    0 & 0
  \end{pmatrix}
  \vec{y} + 
  \begin{pmatrix}
    1\\
    3
  \end{pmatrix}
&
{-}_\A(\vec{x},\vec{y} )
&= \vec{x} + 
  \begin{pmatrix}
    2\\
    3
  \end{pmatrix}
\tpkt
\end{align*}
Further, consider the RMI $\BB$ giving rise to the complexity pair 
$({\geqord{\BB}},{\gord{\BB}})$.
\begin{align*}
\m{0}_\BB &= 
\makebox[0mm][l]{$\m{true}_\BB = \m{false}_\BB = 
\m{\leqslant}_\BB (\vec x, \vec y) = \vec{0}$}
\\
\m{s}_\BB(\vec x) &=
\begin{pmatrix}
  1 & 3\\
  0 & 0 
\end{pmatrix}
\vec{x} +
\begin{pmatrix}
  3 \\
  0
\end{pmatrix} 
&&
{-}_\BB(\vec{x}, \vec{y})
& =
\begin{pmatrix}
  1 & 0\\
  2 & 2 
\end{pmatrix}
\vec{x} +
\begin{pmatrix}
  0 & 0 \\
  1 & 0
\end{pmatrix} 
\\
\m{if_{gcd}}^\sharp_\BB(x,y,z) & =
\begin{pmatrix}
  3 & 0\\
  0 & 0
\end{pmatrix}
\vec y +
\begin{pmatrix}
  3 & 0\\
  0 & 0  
\end{pmatrix}
\vec z 
\\
\m{gcd}^\sharp_\BB(x,y) & =
\makebox[0mm][l]{$
\begin{pmatrix}
  3 & 0\\
  0 & 0
\end{pmatrix}
\vec x +
\begin{pmatrix}
  3 & 0\\
  0 & 0  
\end{pmatrix}
\vec y +
\begin{pmatrix}
  2\\
  0
\end{pmatrix}
$
\tpkt
}
\end{align*}
We obtain 
$\{1,\ldots,5\} \subseteq {\gord{\A}}$,
$\{1,\ldots,5\} \subseteq {\geqord{\BB}}$, and 
$\{18,19,20\} \subseteq {\gord{\BB}}$.

\item
Consider $(\{18,19,20\},\{16\})$.  Note $\UU(\{16\}) = \varnothing$.
By taking the same $\A$ and also $\BB$ as above, we have
$\{1,\ldots,5\} \subseteq {\gord{\A}}$,
$\{1,\ldots,5,18,19,20\} \subseteq {\geqord{\BB}}$, and 
$\{16\} \subseteq {\gord{\BB}}$.
\end{itemize}
Thus, all path constraints are handled by suitably defined RMIs of
dimension 2. Hence, the runtime complexity function of $\RSgcd$ is at most quadratic,
which is unfortunately not optimal, as $\Rc{\RSgcd}$ is linear.
\end{example}

Corollary~\ref{c:dg} is more powerful than Corollary~\ref{c:main}.
We illustrate it with a small example.
\begin{example}
Consider the TRS $\RS$
\begin{align*}
\m{f}(\m{a},\m{s}(x),y) & \to \m{f}(\m{a},x,\m{s}(y)) &
\m{f}(\m{b},x,\m{s}(y)) & \to \m{f}(\m{b},\m{s}(x),y)
\tpkt
\end{align*}
Its weak dependency pairs $\WDP(\RS)$ are
\begin{align*}
1\colon~ \m{f}^\sharp(\m{a},\m{s}(x),y) & \to \m{f}^\sharp(\m{a},x,\m{s}(y)) &
2\colon~ \m{f}^\sharp(\m{b},x,\m{s}(y)) & \to \m{f}^\sharp(\m{b},\m{s}(x),y)
\tpkt
\end{align*}
The corresponding congruence graph consists of the two isolated nodes
$\{1\}$ and $\{2\}$.  It is not difficult to find suitable $1$-dimensional 
RMIs for the nodes, and therefore $\Rc{\RS}(n) = \bO(n)$ is concluded.
On the other hand, it can be verified that the linear runtime complexity
cannot be obtained by Corollary~\ref{c:main} with a $1$-dimensional RMI.
\end{example}

We conclude this section with a brief comparison of the path analysis
developed here and the use of the dependency graph refinement in termination
analysis. First we recall a theorem on 
the dependency graph refinement in conjunction with 
usable rules and innermost rewriting (see~\cite{GAO:2002}, but also~\cite{HirokawaMiddeldorp:2005}).
Similar results hold in the
context of full rewriting, see~\cite{GTSF06,HirokawaMiddeldorp:2007}.
\begin{theorem}[\cite{GAO:2002}]
\label{t:GAO02}
A TRS $\RS$ is innermost terminating if for every maximal cycle 
$\CC$ in the dependency graph $\DG(\RS)$
there exists a reduction pair $(\gtrsim,\succ)$ such that
${\UU(\CC)} \subseteq {\gtrsim}$ and ${\CC} \subseteq {\succ}$.
\end{theorem}

The following example shows that in the context of complexity
analysis it is \emph{not} sufficient to
consider each cycle individually.

\begin{example}[continued from Example~\ref{ex:7}] \label{ex:exp}
Consider the TRS $\RSexp$ introduced in Example~\ref{ex:7}.
\begin{align*}
\m{exp}(\mN) & \to \ms(\mN) & \m{d}(\mN) & \to \mN \\
\m{exp}(\m{r}(x)) & \to \m{d}(\m{exp}(x)) 
& \m{d}(\ms(x)) & \to \ms(\ms(\m{d}(x))) 
\tpkt
\end{align*}
Recall that the (innermost) runtime complexity of $\RSexp$ is exponential.
Let $\PP$ denote the (standard) dependency pairs with respect to
$\RSexp$. Then $\PP$ consists of three pairs:
$1\colon \m{exp}^\sharp(\m{r}(x)) \to \m{d}^\sharp(\m{exp}(x))$,
$2\colon \m{exp}^\sharp(\m{r}(x)) \to \m{exp}^\sharp(x)$, and
$3\colon \m{d}^\sharp(\ms(x)) \to \m{d}^\sharp(x)$.
Hence the dependency graph $\DG(\RSexp)$ contains two 
maximal cycles: $\{2\}$ and $\{3\}$. 

We define two reduction pairs 
$(\geqord{\A},\gord{\A})$ and $(\geqord{\BB},\gord{\BB})$ 
such that the conditions of the theorem are fulfilled.
Let $\A$ and $\BB$ be SLIs such that
$\m{exp}^\sharp_{\A}(x) = x$, $\m{r}_{\A}(x) = x+1$ and
$\m{d}^\sharp_{\BB}(x) = x$, $\ms_{\A}(x) = x+1$. Hence for any
term $t \in \TB$, we have that the derivation heights 
$\dheight(t^\sharp,\rsirew{\{2\}/\UU(\PP)})$ and $\dheight(t^\sharp,\rsirew{\{3\}/\UU(\PP)})$
are linear in $\size{t}$, while $\dheight(t,\rsirew{\RS})$ is (at least) exponential
in $\size{t}$.
\end{example}

Observe that the problem exemplified by Example~\ref{ex:exp}
cannot be circumvented by  replacing the dependency graph employed in 
Theorem~\ref{t:GAO02} with weak (innermost)  dependency graphs. The 
exponential derivation height of terms $t_n$ in Example~\ref{ex:exp} 
is not controlled by the cycles $\{2\}$ or $\{3\}$, but
achieved through the non-cyclic pair $1$ and its usable rules.

Example~\ref{ex:exp} shows an exponential speed-up between the maximal number of 
dependency pair steps within a cycle in the dependency graph and the runtime
complexity of the initial TRS. In the context of derivational complexity this speed-up
may even increase to a primitive recursive function, cf.~\cite{MS:2010}.

While Example~\ref{ex:exp} shows that the usable rules need to be taken
into account fully for any complexity analysis, it is perhaps 
tempting to think that it should suffice to demand that at least one 
weak (innermost) dependency pair in each cycle decreases strictly. 
However this intuition is deceiving as shown by the next example.

\begin{example}
Consider the TRS $\RS$ of
$\m{f}(\m{s}(x),\m{0}) \to \m{f}(x, \m{s}(0))$ and
$\m{f}(x,\m{s}(y)) \to \m{f}(x, y)$.
$\WDP(\RS)$ consists of
$1\colon \m{f}^\sharp(\m{s}(x),\m{0}) \to \m{f}^\sharp(x, \m{s}(x))$
and $2\colon \m{f}^\sharp(x,\m{s}(y)) \to \m{f}^\sharp(x,y)$,
and the weak dependency graph $\WDG(\RS)$ 
contains two cycles $\{1,2\}$ and $\{2\}$.
There are two linear restricted interpretations $\A$ and $\BB$ such that
$\{1,2\} \subseteq {\geqslant_\A} \cup {>_\A}$,
$\{1\} \subseteq {>_\A}$, and
$\{2\} \subseteq {>_\BB}$.
Here, however, we must not conclude linear runtime complexity,
because the runtime complexity of $\RS$ is at least quadratic.
\end{example}

\section{Experiments} \label{Experiments}

All described techniques have been incorporated into
the \emph{Tyrolean Complexity Tool} $\TCT$, 
an open source complexity analyser%
\footnote{Available at \url{http://cl-informatik.uibk.ac.at/software/tct}.}.
The testbed is based on version 8.0.2 of the \emph{Termination Problems Database}
(\emph{TPDB} for short). We consider TRSs without theory annotation, where the
runtime complexity analysis is non-trivial, that is the set of basic terms is
infinite. This testbed comprises 1695 TRSs.
All experiments were conducted on a machine that is identical 
to the official competition server
($8$ AMD Opteron${}^\text{\textregistered}$ 885 dual-core processors
with 2.8GHz, $8\text{x}8$ GB memory). As timeout we use 60 seconds.
The complete experimental data can be found at \url{http://cl-informatik.uibk.ac.at/software/tct/experiments}, where also the testbed employed is detailed.

Table~\ref{tab:1} summarises the experimental results of the here
presented techniques for full runtime complexity analysis in a 
restricted setting. The tests are based on the use of one- and two-dimensional RMIs 
with coefficients over $\{0,1,\ldots, 7\}$
as direct technique (compare Theorem~\ref{t:rmi}) 
as well as in combination with the 
WDP method (compare Corollaries~\ref{c:dp:usable} and~\ref{c:main}) 
and the WDG method (compare Corollary~\ref{c:dg}).
Weak dependency graphs are estimated by the $\TCAP$-based technique (\cite{GTS05}).
The tests indicate the power of the transformation techniques introduced. Note 
that for linear and quadratic runtime complexity the latter techniques are more powerful than the
direct approach. Furthermore note that the WDG method provides overall better bounds
than the WDP method. 

\begin{table}[h]
\centering
\begin{tabular}{@{}l@{}r@{\hspace{1ex}}r@{\hspace{2ex}}r@{\hspace{1ex}}r@{\hspace{2ex}}r@{\hspace{1ex}}r@{}}
  \hline\\[-2ex]
  & \multicolumn{6}{c}{\textit{full}} \\
     result &  direct\,(1) & direct\,(2) &  WDP\,(1) & WDP\,(2) &  WDG\,(1) & WDG\,(2)
  \\[.5ex] \hline \\[-2ex]
  $\OO(1)$   & 16  &  18 &   0 & 0   & 10  & 10 \\
  $\OO(n)$   & 106 & 113 & 123 & 70  & 130 & 67 \\
  $\OO(n^2)$ & 106 & 148 & 123 & 157 & 130 & 158 \\[1ex]
  timeout (60s) & 20 & 88 & 55 & 127 & 103 & 261\\[1ex]
\hline
\end{tabular}
\caption{Experiment results I (one- and two-dimensional RMIs separated)}
\label{tab:1}
\end{table}

However if we consider RMIs upto dimension 3 the picture becomes less
clear, cf.~Table~\ref{tab:2}. Again we compare the direct
approach, the WDP and WDG method and restrict to coefficients over $\{0,1,\ldots, 7\}$. 
Consider for example the test results for cubic runtime complexity with respect to full rewriting. 
While the transformation techniques are still more powerful than the direct approach,
the difference is less significant than in Table~\ref{tab:1}. 
On one hand this is due to the fact that RMIs employing matrices of
dimension $k$ may have a degree strictly smaller
than $k$, compare Theorem~\ref{t:rmi} and
on the other hand note the increase in timeouts for the more advanced techniques. 

Moreover note the seemingly strange behaviour of the WDG method
for innermost rewriting: already for quadratic runtime the WDP method performs better, 
if we only consider the number of yes-instances. This seems to contradict
the fact that the WDG method is in theory more powerful than the WDP method. However,
the explanation is simple: first the sets of yes-instances are incomparable and
second the more advanced technique requires more computation power. 
If we would use (much) longer timeout
the set of yes-instances for WDP would become
a \emph{proper} 
subset of the set
of yes-instances for WDG.
For example the WDG method can prove cubic runtime complexity 
of the TRS \texttt{AProVE\_04/Liveness 6.2} from the TPDB, while the WDP 
method fails to give its bound.

\begin{table}[h]
\centering
\begin{tabular}{@{}l@{}r@{\hspace{1ex}}r@{\hspace{1ex}}r@{\hspace{5ex}}r@{\hspace{1ex}}r@{\hspace{1ex}}r@{}}
\hline \\[-2ex]
& \multicolumn{3}{c}{\textit{full}} 
& \multicolumn{3}{c}{\textit{innermost}}
\\[.5ex]
  result & direct & WDP & WDG & direct & WDP & WDG \\[1ex]
  \hline\\[-1.5ex] 
  $\OO(1)$   &  18  &  0   &  10  & 20  &  0 & 10 \\
  $\OO(n)$   &  135 &  141 &  140 & 135 &  142 & 145  \\
  $\OO(n^2)$ &  161 &  163 &  162 & 173 &  181 & 172 \\
  $\OO(n^3)$ &  163 &  167 &  169 & 179 &  185 & 178 \\[1ex]
  timeout (60s)  &  310 &  459 &  715 & 311 &  458 & 718 \\[1ex]
\hline
\end{tabular}
\caption{Experiment results II ($1\text{--}3$-dimensional RMIs combined)}
\label{tab:2}
\end{table}

In order to assess the advances of this paper in contrast to the conference
versions (see~\cite{HM:2008,HM:2008b}), we present in Table~\ref{tab:3} a comparison between
RMIs with/without the use of usable arguments and a comparison of the WDP or WDG method
with/without the use of the extended weight gap principle. Again we restrict our
attention to full rewriting, as the case for innermost rewriting provides
a similar picture (see~\url{http://cl-informatik.uibk.ac.at/software/tct/experiments} for
the full data).

\begin{table}[h]
\centering
\begin{tabular}{@{}l@{}r@{\hspace{1ex}}r@{\hspace{2ex}}r@{\hspace{1ex}}r@{\hspace{2ex}}r@{\hspace{1ex}}r@{}}
\hline \\[-1.5ex]
& \multicolumn{6}{c}{\textit{full}}
\\[1ex] 
  result 
  & direct\,($-$) & direct\,($+$) 
  & WDP\,($-$) & WDP\,($+$) 
  & WDG\,($-$) & WDG\,($+$) \\[1ex]
  \hline 
  \\[-1.5ex]
  $\OO(1)$   & 4   &  18  & 5   &  0   & 10 &  10  \\
  $\OO(n)$   & 105 &  135 & 102 &  141 & 105 &  140 \\
  $\OO(n^2)$ & 127 &  161 & 118 &  163 & 119 &  162 \\
  $\OO(n^3)$ & 130 &  163 & 120 &  167 & 122 &  169 \\[1ex]
  timeout (60s) & 306 & 310 & 505 &  459 & 655 &  715 \\[1ex]
\hline
\end{tabular}
\caption{Experiment results III ($1\text{--}3$-dimensional RMIs combined)}
\label{tab:3}
\end{table}

Finally, in Table~\ref{tab:4} we present the overall power obtained for the
automated runtime complexity analysis. Here we test the version of
\TCT\ that run for the international annual termination competition
(TERMCOMP)%
\footnote{\url{http://termcomp.uibk.ac.at/termcomp/}.} 
in 2010 in comparison to the most recent version of \TCT\ incorporating
all techniques developed in this paper. In addition we compare with a
recent version of \CAT.%
\footnote{\url{http://cl-informatik.uibk.ac.at/software/cat/}.}

\begin{table}[h]
\centering
\begin{tabular}{@{}l@{\quad}r@{\quad}r@{\quad}r@{\quad}r@{\quad}r@{\quad}r@{}}
\hline\\[-1.5ex]
& \multicolumn{3}{c}{\textit{full}}
& \multicolumn{3}{c}{\textit{innermost}}
\\[.5ex]
  result & \TCT\,(old) & \TCT\,(new) & \CAT & \TCT\,(old) & \TCT\,(new) & \CAT \\[1ex]
  \hline\\[-1.5ex]
  $\OO(1)$   & 10  & 3   & 0 & 10  & 3  &   0 \\
  $\OO(n)$   & 393 & 486 & 439 & 401 & 488 &  439  \\
  $\OO(n^2)$ & 394 & 493 & 452 & 403 & 502 &  452  \\
  $\OO(n^3)$ & 397 & 495 & 453 & 407 & 505 &  453  \\
  $\OO(n^4)$ & 397 & 495 & 454 & 407 & 505 &  454  \\[1ex]
\hline
\end{tabular}
\caption{Experiment results IV ($1\text{--}3$-dimensional RMIs combined)}
\label{tab:4}
\end{table}

The results in Table~\ref{tab:4} clearly show the increase in power in
\TCT, which is due to the fact that the techniques developed in this paper 
have been incorporated.

\section{Conclusion} \label{Conclusion}

In this article we are concerned with automated complexity analysis of
TRSs. More precisely, we establish new and powerful results that allow 
the assessment of polynomial runtime complexity of TRSs fully automatically. 
We established the following results: Adapting techniques from
context-sensitive rewriting, we introduced \emph{usable replacement maps} that
allow to increase the applicability of direct methods. Furthermore we established
the \emph{weak dependency pair method} as a suitable analog of the
dependency pair method in the context of (runtime) complexity analysis. Refinements
of this method have been presented by the use of the \emph{weight gap principle}
and \emph{weak dependency graphs}. 
In the experiments of Section~\ref{Experiments} we assessed the viability
of these techniques.
It is perhaps worthy of note to mention that our motivating examples 
(Examples~\ref{ex:1},~\ref{ex:8}, and~\ref{ex:2}) could not be handled by 
any known technique prior to our results. 

To conclude, we briefly mention related work. 
Based on earlier work by Arai and
the second author (see~\cite{fsttcs:2005}) Avanzini and the second author
introduced \POPSTAR\ a restriction of the recursive path order (RPO) that
induces polynomial innermost runtime complexity (see~\cite{AM:2008,AM:2009}). 
With respect to derivational complexity, Zankl and Korp generalised a simple
variant of our weight gap principle to achieve a modular derivational complexity
analysis (see~\cite{ZK:2010,ZK:2010c}). Neurauter et al.\ refined
in~\cite{NZM:2010} matrix interpretations in the context of derivational complexity 
derivational complexity (see also~\cite{MSW:2008}). Furthermore, Waldmann
studied in~\cite{W:2010} the use of weighted automata in this setting. 
Based on~\cite{HM:2008,HM:2008b} Noschinski et al.\ incorporated a variant 
of weak dependency pairs (not yet published) into the termination prover \Aprove.%
\footnote{This novel version of \Aprove\ (see~\url{http://aprove.informatik.rwth-aachen.de/}) 
for (innermost) runtime complexity took part in TERMCOMP in 2010.}
Currently this method is restricted to innermost runtime complexity, 
but allows for a complexity analysis in the spirit of the dependency pair framework. 
Preliminary evidence suggests that this technique is orthogonal to the methods presented here.
While all mentioned results are concerned with \emph{polynomial} upper bounds on
the derivational or runtime complexity of a rewrite system, Schnabl and the
second author provided in~\cite{MS:2009,MS:2010,MS:2011} an analysis of
the dependency pair method and its framework from a complexity point of view. 
The upshot of this work is that the dependency pair framework may induce
multiple recursive derivational complexity, even if only simple processors
are considered. 

Investigations into the complexity of TRSs are strongly influenced by research in the
field of ICC, which contributed the use of restricted forms of polynomial
interpretations to estimate the complexity, cf.~\cite{BCMT:2001}. Related results
have also been provided in the study of term rewriting characterisations of
complexity classes (compare~\cite{CichonWeiermann:1997}). 
Inspired by Bellantoni and Cook's
recursion theoretic characterisation of
the class of all polynomial time computable functions
in~\cite{BellantoniCook:1992}, Marion~\cite{Marion:2003} defined LMPO, 
a variant of RPO whose compatibility 
with a TRS implies that the functions computed by the TRS is polytime
computable (compare~\cite{CL:1992}). 
A remarkable milestone on this line is the quasi-interpretation
method by Bonfante et al.~\cite{BMM:2009:tcs}. The method
makes use of standard termination methods in conjunction with special polynomial interpretation
to characterise the class of polytime computable functions. In conjunction with \emph{sup-interpretations}
this method is even capable of making use of \emph{standard} dependency pairs (see~\cite{MP:2009}).

In principle we cannot directly compare our result on \emph{polynomial} runtime complexity
of TRSs with the results provided in the setting of ICC: 
the notion of complexity studied is different. However, due to a recent result by Avanzini and the
second author (see~\cite{AM:2010}, but compare also~\cite{LM:2009,LM:2009b}) we know that 
the runtime complexity of a TRS is an \emph{invariant} cost model. Whenever we have polynomial
runtime complexity of a TRS $\RS$, the functions computed by this $\RS$ can be implemented
on a Turing machine that runs in polynomial time. In this context, our results provide
automated techniques that can be (almost directly) employed in the context of ICC. The qualification
only refers to the fact that our results are presented for an abstract form of 
programs, viz.\ rewrite systems.


\begin{thebibliography}{10}
\expandafter\ifx\csname url\endcsname\relax
  \def\url#1{\texttt{#1}}\fi
\expandafter\ifx\csname urlprefix\endcsname\relax\def\urlprefix{URL }\fi
\expandafter\ifx\csname href\endcsname\relax
  \def\href#1#2{#2} \def\path#1{#1}\fi

\bibitem{CKS:1989}
C.~Choppy, S.~Kaplan, M.~Soria, Complexity analysis of term-rewriting systems,
  Theor.~Comput.~Sci. 67~(2--3) (1989) 261--282.

\bibitem{HofbauerLautemann:1989}
D.~Hofbauer, C.~Lautemann, Termination proofs and the length of derivations,
  in: Proc.\ 3rd International Conference on Rewriting Techniques and
  Applications, no. 355 in LNCS, Springer Verlag, 1989, pp. 167--177.

\bibitem{CL:1992}
E.-A. Cichon, P.~Lescanne, Polynomial interpretations and the complexity of
  algorithms, in: Proc.\ 11th International Conference on Automated Deduction,
  Vol. 607 of LNCS, 1992, pp. 139--147.

\bibitem{HM:2008}
N.~Hirokawa, G.~Moser, Automated complexity analysis based on the dependency
  pair method, in: Proc.\ 4th International Joint Conference on Automated
  Reasoning, no. 5195 in LNAI, Springer Verlag, 2008, pp. 364--380.

\bibitem{BMR:2009}
P.~Baillot, J.-Y. Marion, S.~R.~D. Rocca, Guest editorial: Special issue on
  implicit computational complexity, ACM Trans.~Comput.~Log. 10~(4).

\bibitem{ArtsGiesl:2000}
T.~Arts, J.~Giesl, Termination of term rewriting using dependency pairs,
  Theor.~Comput.~Sci. 236 (2000) 133--178.

\bibitem{HM:2008b}
N.~Hirokawa, G.~Moser, Complexity, graphs, and the dependency pair method, in:
  Proc.\ 15th International Conference on Logic for Programming Artificial
  Intelligence and Reasoning, no. 5330 in LNCS, Springer Verlag, 2008, pp.
  652--666.

\bibitem{BaaderNipkow:1998}
F.~Baader, T.~Nipkow, Term {R}ewriting and {A}ll {T}hat, Cambridge University
  Press, 1998.

\bibitem{Terese}
Te{R}e{S}e, Term Rewriting Systems, Vol.~55 of Cambridge Tracks in Theoretical
  Computer Science, Cambridge University Press, 2003.

\bibitem{Geser:1990}
A.~Geser, Relative termination, Ph.D. thesis, Universit{\"a}t Passau (1990).

\bibitem{T07}
R.~Thiemann, The {DP} framework for proving termination of term rewriting,
  Ph.D. thesis, University of Aachen, Department of Computer Science (2007).

\bibitem{EWZ08}
J.~Endrullis, J.~Waldmann, H.~Zantema, Matrix interpretations for proving
  termination of term rewriting, J.~Automated Reasoning 40~(3) (2008) 195--220.

\bibitem{HW06}
D.~Hofbauer, J.~Waldmann, Termination of string rewriting with matrix
  interpretations, in: Proc.\ 17th International Conference on Rewriting
  Techniques and Applications, Vol. 4098 of LNCS, 2006, pp. 328--342.

\bibitem{ArtsGiesl:2001}
T.~Arts, J.~Giesl, A collection of examples for termination of term rewriting
  using dependency pairs, Tech. Rep. AIB-2001-09, RWTH Aachen (2001).

\bibitem{AM:2009}
M.~Avanzini, G.~Moser, Dependency pairs and polynomial path orders, in: Proc.\
  20th International Conference on Rewriting Techniques and Applications, Vol.
  5595 of LNCS, 2009, pp. 48--62.

\bibitem{NZM:2010}
F.~Neurauter, H.~Zankl, A.~Middeldorp, Revisiting matrix interpretations for
  polynomial derivational complexity of term rewriting, in: Proc.\ 17th
  International Conference on Logic for Programming Artificial Intelligence and
  Reasoning, Vol. 6397 of LNCS (ARCoSS), 2010, pp. 550--564.

\bibitem{W:2010}
J.~Waldmann, Polynomially bounded matrix interpretations, in: Proc.\ 21st
  International Conference on Rewriting Techniques and Applications, Vol.~6 of
  LIPIcs, 2010, pp. 357--372.

\bibitem{BCMT:2001}
G.~Bonfante, A.~Cichon, J.-Y. Marion, H.~Touzet, Algorithms with polynomial
  interpretation termination proof, J.~Funct.~Program. 11~(1) (2001) 33--53.

\bibitem{F:2005}
M.~L. Fern\'andez, Relaxing monotonicity for innermost termination,
  Inform.~Proc.~Lett. 93~(1) (2005) 117--123.

\bibitem{GTS05}
J.~Giesl, R.~Thiemann, P.~Schneider-Kamp, Proving and disproving termination of
  higher-order functions, in: Proc.\ 5th International Workshop on Frontiers of
  Combining Systems, 5th International Workshop, Vol. 3717 of LNAI, 2005, pp.
  216--231.

\bibitem{GTSF06}
J.~Giesl, R.~Thiemann, P.~Schneider-Kamp, S.~Falke, Mechanizing and improving
  dependency pairs, J.~Automated Reasoning 37~(3) (2006) 155--203.

\bibitem{HirokawaMiddeldorp:2007}
N.~Hirokawa, A.~Middeldorp, Tyrolean termination tool: Techniques and features,
  Inform.~and Comput. 205 (2007) 474--511.

\bibitem{MS:2010}
G.~Moser, A.~Schnabl, The derivational complexity induced by the dependency
  pair method, Logical Methods in Computer ScienceAccepted for publication.

\bibitem{GAO:2002}
J.~Giesl, T.~Arts, E.~Ohlebusch, Modular termination proofs for rewriting using
  dependency pairs, J.~Symbolic Comput. 34 (2002) 21--58.

\bibitem{HirokawaMiddeldorp:2005}
N.~Hirokawa, A.~Middeldorp, Automating the dependency pair method, Inform.~and
  Comput. 199~(1,2) (2005) 172--199.

\bibitem{fsttcs:2005}
T.~Arai, G.~Moser, Proofs of termination of rewrite systems for polytime
  functions, in: Proc.\ 25th Conference on Foundations of Software Technology
  and Theoretical Computer Science, no. 3821 in LNCS, Springer Verlag, 2005,
  pp. 529--540.

\bibitem{AM:2008}
M.~Avanzini, G.~Moser, Complexity analysis by rewriting, in: Proc.\ 9th
  International Symposium on Functional and Logic Programming, no. 4989 in
  LNCS, Springer Verlag, 2008, pp. 130--146.

\bibitem{ZK:2010}
H.~Zankl, M.~Korp, Modular complexity analysis via relative complexity, in:
  Proc.\ 21st International Conference on Rewriting Techniques and
  Applications, Vol.~6 of LIPIcs, 2010, pp. 385--400.

\bibitem{ZK:2010c}
H.~Zankl, M.~Korp, Modular complexity analysis via relative complexity, Logical
  Methods in Computer ScienceSubmitted.

\bibitem{MSW:2008}
G.~Moser, A.~Schnabl, J.~Waldmann, Complexity analysis of term rewriting based
  on matrix and context dependent interpretations, in: Proc.\ 28th Conference
  on Foundations of Software Technology and Theoretical Computer Science,
  LIPIcs, 2008, pp. 304--315.

\bibitem{MS:2009}
G.~Moser, A.~Schnabl, The derivational complexity induced by the dependency
  pair method, in: Proc.\ 20th International Conference on Rewriting Techniques
  and Applications, Vol. 5595 of LNCS, 2009, pp. 255--269.

\bibitem{MS:2011}
G.~Moser, A.~Schnabl, Termination proofs in the dependency pair framework may
  induce multiply recursive derivational complexities, in: Proc.\ 22nd
  International Conference on Rewriting Techniques and Applications, Vol.~10 of
  LIPIcs, 2011, pp. 235--250.

\bibitem{CichonWeiermann:1997}
E.-A. Cichon, A.~Weiermann, Term rewriting theory for the primitive recursive
  functions., Ann.~Pure Appl.~Logic 83~(3) (1997) 199--223.

\bibitem{BellantoniCook:1992}
S.~Bellantoni, S.~Cook, A new recursion-theoretic characterization of the
  polytime functions, Comput. Complexity 2~(2) (1992) 97--110.

\bibitem{Marion:2003}
J.-Y. Marion, Analysing the implicit complexity of programs, Inform.~and
  Comput. 183 (2003) 2--18.

\bibitem{BMM:2009:tcs}
G.~Bonfante, J.-Y. Marion, J.-Y. Moyen, Quasi-interpretations: A way to control
  resources, Theor.~Comput.~Sci.To appear.

\bibitem{MP:2009}
J.-Y. Marion, R.~P{\'e}choux, Sup-interpretations, a semantic method for static
  analysis of program resources, ACM Trans.~Comput.~Log. 10~(4).

\bibitem{AM:2010}
M.~Avanzini, G.~Moser, Closing the gap between runtime complexity and polytime
  computability, in: Proc.\ 21st International Conference on Rewriting
  Techniques and Applications, Vol.~6 of LIPIcs, 2010, pp. 33--48.

\bibitem{LM:2009}
U.~{Dal Lago}, S.~Martini, On constructor rewrite systems and the
  lambda-calculus, in: Proc.\ 36th ICALP, Vol. 5556 of LNCS, Springer Verlag,
  2009, pp. 163--174.

\bibitem{LM:2009b}
U.~{Dal Lago}, S.~Martini, {D}erivational {C}omplexity is an {I}nvariant {C}ost
  {M}odel, in: Proc.\ 1st FOPARA, 2009.

\end{thebibliography}
\end{document}